\documentclass[12pt]{amsart}

\usepackage[matrix,arrow,curve,frame]{xy}
\usepackage{amsmath,amsthm,amssymb,enumerate}
\usepackage{latexsym}
\usepackage{amscd}
\usepackage[backref]{hyperref}
\usepackage{euscript}
\usepackage{tikz}
\usepackage{tikz-cd}

\usepackage{fullpage}

\usepackage[all]{xypic}

\newtheorem{theorem}{Theorem}[section]

\newtheorem{proposition}[theorem]{Proposition}

\theoremstyle{definition}
\newtheorem{definition}[theorem]{Definition}

\theoremstyle{remark}
\newtheorem{remark}[theorem]{Remark}

\numberwithin{equation}{section}
\newcommand{\beq}{\begin{equation}}
\newcommand{\eeq}{\end{equation}}

\newcommand{\TT}{\mathbb{T}}
\newcommand{\ZZ}{\mathbb{Z}}

\newcommand\bbT{\mathbb T}

\newcommand{\sU}{{\sf U}}

\DeclareMathOperator{\Tr}{Tr}

\renewcommand{\cL}{\mathcal{L}}

\newcommand{\gH}{\mathfrak{H}}

\newcommand{\bu}{\bullet}

\newcommand{\hide}[1]{}

\newcommand {\be}{\begin{equation}}
\newcommand {\ee}{\end{equation}}
\newcommand{\h}{\begin{eqnarray*}}
\newcommand{\e}{\end{eqnarray*}}

\begin{document}

%----------------------------------------------------------------------------------------------------------------------------------------------------------------------------------------------------------%

\title[Loop Hori Formulae for T-duality and Twisted Bismut-Chern Character]
{Loop Hori Formulae for T-duality and Twisted Bismut-Chern Character}


%    author one information
\author{Fei Han}
\address{Department of Mathematics,
National University of Singapore, Singapore 119076}
%\curraddr{}
\email{mathanf@nus.edu.sg}
%\thanks{}

%  author two information
 \author{Varghese Mathai}
\address{School of Mathematical Sciences,
University of Adelaide, Adelaide 5005, Australia}
\email{mathai.varghese@adelaide.edu.au}
%\thanks{}

%    \subjclass is required.
\date{}

\maketitle

\begin{abstract} The main purpose of this paper is to %complete
establish the loop space formulation of T-duality in the presence of background flux. In particular, we construct a loop space analogue of the Hori formula, termed \textbf{the loop Hori map}, and demonstrate that it induces a quasi-isomorphism between the exotic twisted equivariant cohomologies on the free loop spaces of the T-dual sides. Spacetime, when viewed as the constant loops, is a submanifold of loop space. The duality that we prove on loop space restricts to the T-duality  with $H$-flux on spacetime. This significantly refines the earlier work of the authors in \cite{HM15} where T-duality was established after 
localisation to the base space. The construction of the loop Hori map is an application of our generalization of the Bismut--Chern character~\cite{HM15}, originally introduced in the loop space interpretation of the Atiyah--Singer index theorem by Atiyah--Witten~\cite{A85} and Bismut~\cite{B85}.

\end{abstract}

\tableofcontents

%%%%%%%%%%%%%%%%%%%%%%%%%%%%%%%%%%%%%%%%%%%%%%%%%%%%%%%%%%%%%%%%%%%%%%%%%%%%%%%%%%%%%%%%%%%%%%%%%%%%%%%%%%%%%%%%%%%%%%%%%%%%%%
\section*{Introduction}

%----------------------------------------------------------------------------------------------------------------------------------------------

\subsection{Brief Background on T-Duality with \( H \)-Flux and the Hori Map}
In \cite{BEM04a,BEM04b}, T-duality in a background flux was studied {\em mathematically} for the first time. On the \emph{physics} side, the story descends from
Buscher’s rules and their many refinements; see, among others,
\cite{Buscher87,Buscher88,RocekVerlinde92,Alvarez-Intro95,
BergshoeffHullOrtin95,Alvarez2,DuffLuPope98,
GurrieriLouisMicuWaldram03,KachruSchulzTripathyTrivedi03,Hori}
for a representative account of T–duality with background flux.

Given a T-dual pair of circle bundles with connections and $H$-fluxes, 
\begin{center}
\begin{tikzcd}
(Z, A, H)\arrow[rd, "\pi"] & & (\widehat{Z}, \widehat A, \widehat H) \arrow[ld, "\widehat{\pi}"'] \\
 & X & 
\end{tikzcd}
\end{center}
it was shown in \cite{BEM04a, BEM04b} that defined via the correspondence space 
\begin{equation} 
\xymatrix @=6pc @ur { Z \ar[d]_{\pi} &
Z\times_X  \widehat Z, \ar[d]_{\widehat p} \ar[l]^{p} \\
 X & \widehat Z \ar[l]^{\widehat \pi}}
\end{equation}
the {\bf Hori map}, 
\begin{align}
T_*: (\Omega^{\bullet}(Z)^{\TT}, d+H) &\to (\Omega^{\bullet+1}(\widehat{Z})^{\widehat \TT}, -(d+\widehat{H}))\\
\omega \quad &\mapsto \int_{(Z\times_X  \widehat Z)/\widehat Z}\omega \wedge e^{-A\wedge \widehat{A}  }, \label{Horiexp}
\end{align}
is a chain map isomorphism between the twisted, $\mathbb{Z}_2$-graded complexes. Here, $\bullet$ indicates the parity, i.e., whether the degree is even or odd. One can also define the dual Hori map
\begin{align*}
\widehat T_*: (\Omega^{\bullet}(\widehat Z)^{\widehat \TT}, d+\widehat H) &\to (\Omega^{\bullet+1}(Z)^{\TT}, -(d+H))\\
\omega \quad &\mapsto \int_{(Z\times_X  \widehat Z)/ Z}\omega \wedge e^{A\wedge\widehat{A}},
\end{align*} and T-duality can be stated as
\be \label{tdual} T\circ \widehat T=-Id, \ \ \widehat T\circ T=-Id.\ee In particular, this induces an isomorphism on the twisted cohomology:
\begin{align*}
T: H^{\bullet}_{d+H}(Z) &\to H^{\bullet+1}_{d+\widehat{H}}(\widehat{Z}).
\end{align*}
The Chern class is exchanged with the $H$-flux. So in general the topologies of $Z$ and the T-dual $\widehat Z$ are different.

{On the other hand, T-duality in string theory can be realised as a transformation acting on the
worldsheet fields in the 2-dimensional nonlinear sigma model \cite{HLSUZ}.
In \cite{KL, HLSUZ2}, T-duality is studied with supersymmetry in the target space of the sigma model.  This
realisation is straightforward when there is an abelian isometry of the target
space which at the same time is a symmetry of the worldsheet action.
One finds the transformation on the worldsheet fields as
well as Buscher's rules \cite{Buscher} which yield the dual
background. The duality transformation is symmetric in the sense that
one can start from either theory (the original or the dual one) and
obtain the other one via the same gauging process \cite{Alvarez2}. A relevant account of the sigma model is in \cite{W88}.}

%----------------------------------------------------------------------------------------------------------------------------------------------
\subsection{Preliminary Study of the Loop Space Perspective of T-Duality with \( H \)-Flux}
In this paper, we complete the loop space formulation of T-duality in the presence of a background \( H \)-flux, continuing the program initiated in~\cite{HM15}. Specifically, we explore the behavior of T-duality at the level of free loop spaces by looping the classical T-duality diagram, i.e looking at the following diagram:
\begin{center}
\begin{tikzcd}
LZ \arrow[rd, "L\pi"] & & L\widehat{Z} \arrow[ld, "L\widehat{\pi}"'] \\
 & LX & 
\end{tikzcd}
\end{center}
This loop space framework offers the geometric side of T-duality from the perspective of the \emph{1-dimensional nonlinear sigma model with \( H \)-flux}, revealing deeper structural connections between geometry, topology, and string theory.

In~\cite{HM15}, we introduced the \emph{exotic twisted \(S^1 \)-equivariant cohomology} for the loop space \( LM \) of a smooth manifold \( M \) equipped with gerbe data 
\[
(H, B_\alpha, F_{\alpha\beta}, (L_{\alpha\beta}, \nabla^{L_{\alpha\beta}})).
\]
This exotic cohomology is defined using \( S^1 \)-invariant differential forms on \( LM \) with coefficients in the holonomy line bundle \( \mathcal{L}^B \) associated to the gerbe with connection. The differential in this complex is given by the \emph{equivariantly flat superconnection}
\[
\nabla^{\mathcal{L}^B} - i_K + \overline{H},
\]
in the sense of~\cite{MQ, Q}, where \( K \) is the rotation vector field on \( LM \), and \( \overline{H} \) is a degree-3 \( S^1 \)-invariant differential form on \( LM \) canonically determined by the 3-form curvature \( H \) of the gerbe.

We proved that the \emph{completed periodic version} of this cohomology theory,
\[
h^\bullet_{S^1}(LM, \mathcal{L}^B; \overline{H}) := H\left( \Omega^\bullet(LM, \mathcal{L}^B)^{S^1}[[u, u^{-1}]],\ \nabla^{\mathcal{L}^B} - u i_K + u^{-1} \overline{H} \right),
\]
has a Borel–Witten type \emph{localization theorem}, which identifies this cohomology with the twisted de Rham cohomology of the fixed point submanifold of \( LM \), namely the original manifold \( M \) (see Theorem~\ref{local}):
\[
h^\bullet_{S^1}(LM, \nabla^{\mathcal{L}^B}; H) \overset{\mathrm{res}}{\cong} H^\bullet\left( \Omega(M)[[u, u^{-1}]],\ d + u^{-1} H \right) \cong H^\bullet(M, H)[[u, u^{-1}]].
\]

Applying this localization to both sides of the T-duality correspondence, we obtain the following commutative diagram:

\begin{equation}
\begin{gathered} \label{localT}
\xymatrix@C=4.5em@R=3em{
h_{L\TT, \, S^1}^\bullet(LZ, \mathcal{L}^B; \overline{H}) 
\ar[d]_{\mathrm{res}}^{\cong}  
\ar@{..>}[r]^{\ ?\ } &
h_{L\widehat \TT, \, S^1}^{\bullet+1}(L\widehat{Z}, \mathcal{L}^{\widehat{B}}; \overline{\widehat{H}}) 
\ar[d]_{\mathrm{res}}^{\cong} \\
H^\bullet(\Omega(Z)^{\mathbb{T}}[[u, u^{-1}]],\ d + u^{-1} H) 
\ar[r]_{T}^{\cong} & 
H^{\bullet+1}(\Omega(\widehat{Z})^{\widehat{\mathbb{T}}}[[u, u^{-1}]],\ d + u^{-1} \widehat{H})
}
\end{gathered}
\end{equation}

\noindent

Here, for simplicity, we continue to use the symbol \( T \), even though it denotes a slight modification of the classical Hori transform~\eqref{Horiexp}, in which the kernel \( e^{-A\wedge \widehat{A}}   \) is replaced by \( e^{-u A\wedge \widehat{A}  } \). The notation \( h_{L\mathbb{T},\, S^1} \) and \( h_{L\widehat{\mathbb{T}},\, S^1} \) refers to the use of \( L\mathbb{T} \)-invariant and \( L\widehat{\mathbb{T}} \)-invariant forms, respectively, in the exotic twisted equivariant cohomology theory. This refinement does not affect the localization process.

This diagram (\ref{localT})  illustrates that, after localization, a duality exists between the loop space of \( Z \) and its T-dual counterpart, the loop space of \( \widehat{Z} \). \textbf{To fully capture the loop space perspective of T-duality, it is necessary to construct a \emph{direct} formulation of the Hori transform from \( LZ \) to \( L\widehat{Z} \), without relying on the restriction localization to the constant loops. }

%----------------------------------------------------------------------------------------------------------------------------------------------
\subsection{The Loop Hori Map and T-Duality on Loop Spaces: Main Results} 
In what follows, we present our approach to this problem. The construction, the loop space analogue of the Hori formula, termed \textbf{the loop Hori map}, is the central objective of this paper. We achieve this by an application of our generalization of the Bismut--Chern character~\cite{HM15}, originally introduced in the loop space interpretation of the Atiyah--Singer index theorem by Atiyah--Witten~\cite{A85} and Bismut~\cite{B85},  and by defining an ``integration along the fiber'' operation in the loop space of circle bundles, which constitutes another key novelty of this paper. 

{\bf We will see that the loop Hori map, together with the related duality isomorphisms on the loop spaces, makes the T-duality between the spacetimes \( Z \) and \( \widehat{Z} \) appear as a shadow of theories formulated on their loop spaces.}

$\, $

Let \( \mathcal{L}_h^B \) denote the \textbf{horizontal holonomy line bundle} with connection on \( LZ \), associated to the gerbe with connection \( (H, B_\alpha, F_{\alpha\beta}, (L_{\alpha\beta}, \nabla^{L_{\alpha\beta}})) \) on \( Z \). By \emph{horizontal}, we mean that the transition functions are built using the \textbf{horizontal holonomy} \( \mathrm{hol}_h(\nabla^{L_{\alpha\beta}}) \) defined with respect to the horizontal vector field \( K^h \) along loops. One has a horizontal analogue of Theorem~\ref{flat}:
\[
\left( \nabla^{\mathcal{L}_h^B} - i_{K^h} + \overline{H} \right)^2 + L_{K^h} = 0.
\]
For the horizontal theory, one still has a localization isomorphism (see Theorem \ref{local}):
\[
H\left( \Omega^\bullet(LZ, \mathcal{L}^B_h)^{L\mathbb{T}}[[u, u^{-1}]],\ \nabla^{\mathcal{L}_h^B} - u i_{K^h} + u^{-1} \overline{H} \right) 
\underset{\cong}{\overset{\mathrm{res}}{\longrightarrow}} 
H^\bullet\left( \Omega(Z)^{\mathbb{T}}[[u, u^{-1}]],\ d + u^{-1} H \right).
\]

Similarly, on the dual side \( \widehat{Z} \), one has the \textbf{horizontal holonomy line bundle} \( \mathcal{L}_h^{\widehat{B}} \) with connection \( \nabla^{\mathcal{L}_h^{\widehat{B}}} \), arising from the gerbe with connection \( (\widehat{H}, \widehat{B}_\alpha, \widehat{F}_{\alpha\beta}, (\widehat{L}_{\alpha\beta}, \nabla^{\widehat{L}_{\alpha\beta}})) \) on \( \widehat{Z} \). This bundle satisfies the horizontal flatness condition:
\[
\left( \nabla^{\mathcal{L}_h^{\widehat{B}}} - i_{\widehat{K}^h} + \overline{\widehat{H}} \right)^2 + L_{\widehat{K}^h} = 0.
\]
Moreover, one has a parallel localization isomorphism on the dual side:
\[
H\left( \Omega^\bullet(L\widehat{Z}, \mathcal{L}_h^{\widehat{B}})^{L\widehat{\mathbb{T}}}[[u, u^{-1}]],\ \nabla^{\mathcal{L}_h^{\widehat{B}}} - u i_{\widehat{K}^h} + u^{-1} \overline{\widehat{H}} \right) 
\underset{\cong}{\overset{\mathrm{res}}{\longrightarrow}} 
H^\bullet\left( \Omega(\widehat{Z})^{\widehat{\mathbb{T}}}[[u, u^{-1}]],\ d + u^{-1} \widehat{H} \right).
\]

{\em Our loop Hori map \( LT_* \) will be a morphism of the form}
\[
LT_*: \Omega^\bullet(LZ, \mathcal{L}^B_h)^{L\mathbb{T}}[[u, u^{-1}]] 
\longrightarrow 
\Omega^{\bullet+1}(L\widehat{Z}, \mathcal{L}_h^{\widehat{B}})^{L\widehat{\mathbb{T}}}[[u, u^{-1}]].
\]
It is constructed via the \emph{loop correspondence space}
\begin{equation} 
\xymatrix @=6pc @ur { LZ \ar[d]_{L\pi} &
LZ\times_{LX}  L\widehat Z, \ar[d]_{L\widehat p} \ar[l]^{Lp} \\
 LX & L\widehat Z \ar[l]^{L\widehat \pi}}
\end{equation}

In light of the original Hori transform~\eqref{Horiexp}, the construction of the loop Hori map \( LT_* \) hinges on addressing two key questions: \newline
(i) What is the appropriate analogue of the kernel \( e^{-u A\wedge \widehat{A} } \) in the loop space setting? \newline
(ii) What replaces the integration along the fiber \( \int_{(Z \times_X \widehat{Z})/\widehat Z} \) when working at the level of loop spaces?

We address these questions as follows:

\noindent (i) We show that the appropriate loop space analogue of the kernel \( e^{-u A\wedge \widehat{A}} \) is given by the \textbf{horizontal twisted Bismut--Chern character} of the \textbf{Poincaré gerbe module}. This character arises from the geometry of the correspondence space \( Z \times_X \widehat{Z} \), equipped with a flat gerbe structure. The horizontal twisted Bismut--Chern character is a differential form on the loop correspondence space \( LZ \times_{LX} L\widehat{Z} \), valued in the tensor product of horizontal holonomy line bundles, and plays the role of the integral kernel in the loop Hori map.

{Let us elaborate. As reviewed in Section~\ref{reviewT} and indicated in Section~\ref{corrgerbe}, one has the gerbe data over the space \( Z \) \[
(H, B_\alpha, F_{\alpha\beta}, (L_{\alpha\beta}, \nabla^{L_{\alpha\beta}})),
\]
and on the dual space \( \widehat{Z} \), the corresponding dual gerbe data 
\[
(\widehat{H}, \widehat{B}_\alpha, \widehat{F}_{\alpha\beta}, (\widehat{L}_{\alpha\beta}, \nabla^{\widehat{L}_{\alpha\beta}})).
\]
\noindent\emph{In the present context, we work over the base manifold \( X \) using a Brylinski-type cover \( \{U_\alpha\} \), which is particularly suited for constructions involving loop space. Despite this variation, all arguments from Section~\ref{reviewT} for T-duality remain valid without essential modification.}}

{The correspondence space \( Z \times_X \widehat{Z} \) is locally covered by open sets of the form
\[
U_\alpha \times \mathbb{T} \times \widehat{\mathbb{T}}.
\]
On the overlaps
\[
U_{\alpha\beta} \times \mathbb{T} \times \widehat{\mathbb{T}} = (U_\alpha \cap U_\beta) \times \mathbb{T} \times \widehat{\mathbb{T}},
\]
we consider
\[
\left( \widehat{p}^* \widehat{L}_{\alpha\beta} \otimes \left( p^* L_{\alpha\beta} \right)^{-1},\ 
1 \otimes \widehat{p}^* \nabla^{\widehat{L}_{\alpha\beta}} - p^* \nabla^{L_{\alpha\beta}} \otimes 1 \right).
\]
which is a gerbe $ \mathcal{CG} $ with connection, referred to as the \textbf{correspondence gerbe with connection},
 whose curvature is \( H - \widehat{H} \).} The horizontal holonomy line bundle of this gerbe is 
$
(Lp)^*\mathcal{L}^B_h)^{-1} \otimes (L\widehat{p})^*\mathcal{L}_h^{\widehat{B}}
$
over the loop correspondence space \( LZ \times_{LX} L\widehat{Z} \). The diagram below summarizes this structure:

{\footnotesize
\begin{equation}
\label{eq:T-duality-loop}
\begin{gathered}
\xymatrix@C=4em@R=3em{
& ((Lp)^*\mathcal{L}^B_h)^{-1} \otimes (L\widehat{p})^*\mathcal{L}_h^{\widehat{B}} \ar[d] & \\
(\mathcal{L}^B_h, \nabla^{\mathcal{L}_h^B}) \ar[d] 
& (LZ \times_{LX} L\widehat{Z},\ \overline H - \overline{\widehat H}) \ar[dl]_{Lp} \ar[rd]^{L\widehat{p}} 
& (\mathcal{L}_h^{\widehat{B}}, \nabla^{\mathcal{L}_h^{\widehat{B}}}) \ar[d] \\
(LZ, \overline{H}) \ar[rd]_{L\pi} 
& & (L\widehat{Z}, \overline{\widehat{H}}) \ar[ld]^{L\widehat{\pi}} \\
& LX &
}
\end{gathered}
\end{equation}
}

Meanwhile, the space \( Z \times_X \widehat{Z} \) is a \( \TT \times \widehat{\TT} \)-bundle over the base \( X \), and on each fiber \( \TT \times \widehat{\TT} \), there exists a \textbf{Poincaré line bundle}. Although these line bundles do not globally patch together to form a single line bundle over \( Z \times_X \widehat{Z} \), we will show that they assemble into a gerbe module \( \mathcal{P} \) over the correspondence gerbe \( \mathcal{CG} \). This gerbe module carries a canonical connection \( \nabla^\mathcal{P} = \{\nabla^{\mathcal{P}_\alpha}\} \), induced from the standard antisymmetric connection on the Poincaré line bundles. Moreover, we demonstrate the existence of local 1-forms \( \Theta_\alpha \) on the base \( X \) such that the modified connection \( \nabla^{\mathcal{P}_\alpha} + \Theta_\alpha \) forms a gerbe module connection compatible with the connection on \( \mathcal{CG} \). Note that $\{d\Theta_\alpha\}$ patch together to be a closed 2 form on $X$, which we denote by $\Lambda$. 

We then take the horizontal twisted Bismut--Chern character of this gerbe module:
\[
\mathrm{BCh}^h_{\widehat{H} - H}(\mathcal{P}, \nabla^{\mathcal{P}} + \Theta) 
\in \Omega^\bullet\left(LZ \times_{LX} L\widehat{Z},\ ((Lp)^*\mathcal{L}^B_h)^{-1} \otimes (L\widehat{p})^*\mathcal{L}_h^{\widehat{B}} \right)^{L\mathbb{T} \times L\widehat{\mathbb{T}},\ K^h},
\]
where the right-hand side denotes the space of \( (L\mathbb{T} \times L\widehat{\mathbb{T}}) \)-invariant differential forms with coefficients in the indicated bundle, and satisfying \( L_{K^h}\omega = 0 \).

This form, \( \mathrm{BCh}^h_{\widehat{H} - H}(\mathcal{P}, \nabla^{\mathcal{P}} + \Theta) \), plays the role of the kernel \( e^{-u \widehat{A} \wedge A} \) in the definition of the loop Hori map. We will see that \( \mathrm{BCh}^h_{\widehat{H} - H}(\mathcal{P}, \nabla^{\mathcal{P}} + \Theta) \) reduces to the twisted Chern character on the constant loop space of \( LZ \times_{LX} L\widehat{Z} \), namely, the correspondence space \( Z \times_X \widehat{Z} \):
\begin{equation}
\mathrm{Ch}_{\widehat{H} - H}(\mathcal{P}, \nabla^{\mathcal{P}} + \Theta) = e^{-(A \wedge \widehat{A} + \Lambda)},
\end{equation}
thus generalizing the Chern character of the Poincaré line bundle with connection in the standard case. Recall that \( A \) and \( \widehat{A} \) are the connections on the dual sides, and compare this kernel with the one appearing in (\ref{Horiexp}).

There is a dual version
\[
\mathrm{BCh}^h_{H - \widehat{H}}(\mathcal{P}^{-1}, \nabla^{\mathcal{P}^{-1}} - \Theta) 
\in \Omega^\bullet\left( LZ \times_{LX} L\widehat{Z},\, (Lp)^*\mathcal{L}_h^B \otimes \left((L\widehat{p})^* \mathcal{L}_h^{\widehat{B}}\right)^{-1} \right)^{L\mathbb{T} \times L\widehat{\mathbb{T}}, K^h},
\]
which will be used to construct the dual loop Hori map.

Details of this construction will be provided in Section~\ref{Poincare}.

$\, $

\noindent
(ii) Secondly, we need to find an appropriate analogue of the fiber integration operation \( \int_{(Z \times_X \widehat{Z})/Z} \) appearing in~\eqref{Horiexp}. Recall that
\[
\int_{(Z \times_X \widehat{Z})/Z} \omega = i_v \omega,
\]
where \( \omega \in \Omega^\bullet(Z)^{\mathbb{T}} \), and \( v \) denotes the vertical invariant vector field associated with the principal circle bundle \( \mathbb{T} \to Z \to X \).

For \( \omega \in \Omega^*(LZ)^{L\mathbb{T}} \), we define
\[
\int_{LZ/LX} \omega := i_{v_L} \omega,
\]
where \( v_L \) is the vector field on \( LZ \) induced from \( v \), defined pointwise by \( v_L|_{\gamma(t)} = v|_{\gamma(t)} \) for any loop \( \gamma \in LZ \). We refer to the operation \( \int_{LZ/LX} \) as the \textbf{integration along the fiber in the loop space of the circle bundle}.

Section~\ref{integration} is devoted to the justification of this definition and the terminology. We begin by reviewing Bismut's family version of the Berline--Vergne~\cite{BV} and Duistermaat--Heckman~\cite{DH82} localization formulas for fiber integration in finite-dimensional geometry. We then discuss how this family localization formula can be formally extended to the loop bundle \( LF \to LM \to LB \) associated with a fiber bundle \( F \to M \to B \), leading to a loop space analogue of the family Atiyah--Witten localization formula (see formula~\eqref{loopintegration}). Finally, in the last subsection, we specialize to the case of a principal circle bundle and introduce the operation \( \int_{LZ/LX} \). We demonstrate that this operation agrees with the formal family Atiyah--Witten localization formula introduced in Section~\ref{loopbundle}, and we investigate several of its important properties.

$\, $

After all these preparations, we are ready to construct the \textbf{loop Hori map} and establish loop space T-duality in the presence of background \( H \)-flux (see Section~\ref{LoopHori} for full details). 

The \textbf{loop Hori map} 
\[
LT_*: \Omega^\bullet(LZ, \mathcal{L}^B_h)^{L\mathbb{T}}[[u, u^{-1}]] \longrightarrow \Omega^{\bullet+1}(L\widehat{Z}, \mathcal{L}_h^{\widehat{B}})^{L\widehat{\mathbb{T}}}[[u, u^{-1}]]
\]
is defined by
\[
LT_* G := \int_{LZ \times_{LX} L\widehat{Z} / L\widehat{Z}} 
BCh^h_{\widehat{H}-H}(\mathcal{P}, \nabla^{\mathcal{P}}+\Theta) \widehat{\otimes} (Lp)^*G.
\]
The \textbf{dual loop Hori map}
\[
L\widehat{T}_*: \Omega^\bullet(L\widehat{Z}, \mathcal{L}_h^{\widehat{B}})^{L\widehat{\mathbb{T}}}[[u, u^{-1}]] 
\longrightarrow \Omega^{\bullet+1}(LZ, \mathcal{L}^B_h)^{L\mathbb{T}}[[u, u^{-1}]]
\]
is defined by
\[
L\widehat{T}_* \widehat{G} := \int_{LZ \times_{LX} L\widehat{Z} / LZ} 
BCh^h_{H - \widehat{H}}(\mathcal{P}^{-1}, \nabla^{\mathcal{P}^{-1}} - \Theta) \widehat{\otimes} (L\widehat{p})^* \widehat{G}.
\]

$\, $

Our results for loop space T-duality are \newline

\noindent  {\bf Theorem A} (Theorem \ref{main1}) \leavevmode
\begin{itemize}
    \item[(i)] The loop Hori map and its dual are inverses up to a sign:
    \[
    L\widehat{T}_* \circ LT_* = -\mathrm{Id}, \quad LT_* \circ L\widehat{T}_* = -\mathrm{Id}.
    \]
    
    \item[(ii)] Both maps are chain maps with respect to the twisted differentials. Specifically,
    \begin{align*}
        &\left( \nabla^{\mathcal{L}_h^{\widehat{B}}} - u\, i_{\widehat{K}^h} + u^{-1} \overline{\widehat{H}} \right) \circ LT_*
        = LT_* \circ \left( \nabla^{\mathcal{L}_h^B} - u\, i_{K^h} + u^{-1} \overline{H} \right), \\
        &\left( \nabla^{\mathcal{L}_h^B} - u\, i_{K^h} + u^{-1} \overline{H} \right) \circ L\widehat{T}_*
        = L\widehat{T}_* \circ \left( \nabla^{\mathcal{L}_h^{\widehat{B}}} - u\, i_{\widehat{K}^h} + u^{-1} \overline{\widehat{H}} \right).
    \end{align*}
\end{itemize}

$\, $

As a consequence, we have \newline

\noindent  {\bf Theorem B} (Theorem \ref{main2}) The loop Hori map \newline
\begin{footnotesize}
\[
LT_* \colon 
\left(
\Omega^\bullet(LZ, \mathcal{L}^B_h)^{L\mathbb{T}, K^h}[[u, u^{-1}]],\ 
\nabla^{\mathcal{L}_h^B} - u\, i_{K^h} + u^{-1} \overline{H}
\right)
\longrightarrow 
\left(
\Omega^{\bullet+1}(L\widehat{Z}, \mathcal{L}_h^{\widehat{B}})^{L\widehat{\mathbb{T}}, \widehat{K}^h}[[u, u^{-1}]],\ 
\nabla^{\mathcal{L}_h^{\widehat{B}}} - u\, i_{\widehat{K}^h} + u^{-1} \overline{\widehat{H}}
\right)
\]
\end{footnotesize}
is a quasi-isomorphism.

Let
\[
h_{L\mathbb{T}, K^h}^\bullet(LZ, \mathcal{L}^B_h; \overline{H}) := H\left(
\Omega^\bullet(LZ, \mathcal{L}^B_h)^{L\mathbb{T}, K^h}[[u, u^{-1}]],\ 
\nabla^{\mathcal{L}^B_h} - u i_{K^h} + u^{-1} \overline{H}
\right)
\]
denote the cohomology of the horizontal exotic twisted \( L\mathbb{T} \)-equivariant complex.

By Theorem~\ref{local}, the restriction to constant loops induces an isomorphism:
\[
\mathrm{res}: h_{L\mathbb{T}, K^h}^\bullet(LZ, \mathcal{L}^B_h; \overline{H}) 
\xrightarrow[\cong]{\ \ \ \mathrm{res} \ \ \ } 
H^\bullet(\Omega(Z)^{\mathbb{T}}[[u, u^{-1}]],\ d + u^{-1}H),
\]
and similarly on the dual side.

\vspace{1em}

Restricting the loop Hori map to the constant loop spaces, we obtain a slight modification of the classical Hori map~\eqref{eqn:Hori}. The \textbf{modified Hori map}
\[
T_*': \Omega^\bullet(Z)^{\mathbb{T}}[[u, u^{-1}]] \longrightarrow \Omega^{\bullet+1}(\widehat{Z})^{\widehat{\mathbb{T}}}[[u, u^{-1}]]
\]
is defined by
\[
T_*'(G) := \int_{Z \times_X \widehat{Z} / \widehat{Z}} e^{-u(A \wedge \widehat{A} + \Lambda)} \cdot p^*G.
\]
The \textbf{dual modified Hori map}
\[
\widehat{T}_*': \Omega^\bullet(\widehat{Z})^{\widehat{\mathbb{T}}}[[u, u^{-1}]] \longrightarrow \Omega^{\bullet+1}(Z)^{\mathbb{T}}[[u, u^{-1}]]
\]
is defined by
\[
\widehat{T}_*'(\widehat{G}) := \int_{\widehat{Z} \times_X Z / Z} e^{u(A \wedge \widehat{A} + \Lambda)} \cdot \widehat{p}^*(\widehat{G}).
\]

\noindent
Since \( \Lambda \) is a closed 2-form on \( X \), it follows that the modified Hori maps \( T_*' \) and \( \widehat{T}_*' \) retain the same formal properties as the original Hori maps described in Theorem~\ref{thm:T-duality}.

$\, $

\noindent{\bf Theorem C} (Theorem \ref{main3}) There is a commutative diagram:
\begin{equation}
\begin{gathered}
\xymatrix@C=4.5em@R=3em{
h_{L\mathbb{T}, K^h}^\bullet(LZ, \mathcal{L}^B_h; \overline{H}) 
\ar[d]_{\mathrm{res}}^{\cong}  
\ar[r]^{LT} & 
h_{L\widehat{\mathbb{T}}, \widehat{K}^h}^{\bullet+1}(L\widehat{Z}, \mathcal{L}_h^{\widehat{B}}; \overline{\widehat{H}}) 
\ar[d]_{\mathrm{res}}^{\cong} \\
H^\bullet(\Omega(Z)^{\mathbb{T}}[[u, u^{-1}]], d + u^{-1}H) 
\ar[r]_{T'} & 
H^{\bullet+1}(\Omega(\widehat{Z})^{\widehat{\mathbb{T}}}[[u, u^{-1}]], d + u^{-1} \widehat{H})
}
\end{gathered}
\end{equation}

Theorem A, B and C represent our solutions to the loop space T-duality with $H$-flux. In the bulk of the paper, Section \ref{graded}, we give a further generalization of the loop Hori formula, namely the {\bf graded loop Hori formula}. 

A paper of Belov, Hull and Minasian treating the double field theoretic approach to T-duality is \cite{BHM}. 
In §4 of that work T-duality is recast in the language of canonical quantization of the sigma–model phase space, 
thereby extending the perspective of \cite{Alvarez-Intro95}.

Adopting the notation of our paper, let $Z \to X$ be a principal circle bundle with fibre $S^1$. 
Consider the holonomy line bundle $\mathcal L_H \to LZ$ associated to a gerbe with connection on $Z$, 
and pull it back to $T^*LZ$. The connection on $\mathcal L_H$ has curvature a closed $2$-form on $LZ$, 
which we again pull back to $T^*LZ$. Twisting the standard symplectic form by this curvature yields a 
twisted symplectic structure $\omega_Z$ on $T^*LZ$.

Apply the same construction to the correspondence space 
\[
Y \;=\; Z \times_X \widehat Z,
\]
equipped with closed $3$-form $H - \widehat H$. One obtains a twisted symplectic form $\omega_Y$ on $T^*LY$, 
invariant under the action of $\widehat S^1 \times S^1$, and similarly a twisted form $\omega_{\widehat Z}$ on $T^*L\widehat Z$.

The main structural statement, Theorem~4.1 of \cite{BHM}, is that
the symplectic reduction of the Hamiltonian space $(T^*LY,\omega_Y)$ by the $\widehat S^1$–action yields 
$(T^*LZ,\omega_Z)$, whereas reduction by the $S^1$–action yields $(T^*L\widehat Z,\omega_{\widehat Z})$.
In particular, for sigma–model actions on circle bundles essentially governed by the symplectic form, 
\cite{BHM} exhibits T-duality as a symmetry of a line bundle on $T^*LY$.

Our viewpoint is complementary: we study D-brane charges and their transformation under the 
extension of T-duality to loop space, rather than the line-bundle symmetry on the doubled phase space.

$\, $

\noindent{\bf Acknowledgements.} Fei Han was partially supported by the grant AcRF A-8000451-00-00 from National University of Singapore. Varghese Mathai was supported by funding from the Australian Research Council, through the Australian Laureate Fellowship FL170100020. V.M presented these results at the recent conference, {\em BV-formalism and quantum field theories,} at the Mittag-Leffler Institute, Stockholm. He thanks 
Maxim Zabzine for positive feedback.

%%%%%%%%%%%%%%%%%%%%%%%%%%%%%%%%%%%%%%%%%%%%%%%%%%%%%%%%%%%%%%%

\section{Integration Along the Fiber in Loop Spaces}\label{integration}

The purpose of this section is to introduce a version of integration along the fiber for looped principal circle bundles, which will play a key role in formulating the loop Hori formula for the loop space perspective of T-duality.

In Section~\ref{finite}, we begin by reviewing the family version of the Berline--Vergne~\cite{BV} and Duistermaat--Heckman~\cite{DH82} localization formula for integration along the fiber, as developed by Bismut~\cite{B86} in the finite-dimensional setting.

In Section~\ref{loopbundle}, we consider a fiber bundle \( F \to M \to B \), and study integration along the fibers of the associated loop bundle \( LF \to LM \to LB \). We demonstrate how, when the original bundle is equipped with a flat connection, one can formally apply the finite-dimensional family localization results from Section~\ref{finite} to this looped setting, thereby obtaining a version of the family Atiyah--Witten localization formula on free loop spaces (see formula~\eqref{loopintegration}).

In Section~\ref{circle bundle}, we focus on the special case of a principal circle bundle \( \TT \to M \to B \). For the corresponding loop bundle \( L\TT \to LM \to LB \), we define a rigorous notion of fiber integration \( \int_{LM/LB} \omega \) for an \( L\TT \)-invariant form \( \omega \) on \( LM \) (see Definition \ref{intfiberloop}). We demonstrate that this operation agrees with the formal family Atiyah--Witten localization formula introduced in Section~\ref{loopbundle}, thereby justifying the use of the notation \( \int_{LM/LB} \). We refer to this construction as \emph{integration along the fiber in the looped circle bundles}.

%----------------------------------------------------------------------------------------------------------------------------------------------

\subsection{Family Localization in Finite Dimension} \label{finite} We start from the case of a general fiber bundle.

Let \( \pi: M \to B \) be a smooth fiber bundle with compact, connected, oriented fiber \( C \), a smooth manifold. We equip this bundle with a set of geometric data as described in~\cite{B86, B86I}.

Suppose the fiber bundle is endowed with a connection, i.e., a choice of a \textbf{horizontal subbundle} \( T^H M \subset TM \) of the tangent bundle, such that:
\begin{equation} \label{split}
    TM = T^H M \oplus TC,
\end{equation}
where \( TC \) denotes the \textbf{vertical tangent bundle}. In particular, for each \( y \in B \), the differential \( \pi_* \) induces an isomorphism between \( T^H M \) and \( T_{\pi(y)}B \).

Let \( g^{TB} \) be a Riemannian metric on \( B \), which lifts to a scalar product on \( T^H M \). We also assume \( TC \) is equipped with a scalar product \( g^{TC} \). The scalar products on \( T^H M \) and \( TC \) together define a Riemannian metric on \( TM \), by declaring \( T^H M \perp TC \). Hence, \( M \) becomes a Riemannian manifold.

For \( x \in M \) and \( Y \in T_{\pi(x)}B \), we denote by \( Y^H \in T^H M \) the \textit{horizontal lift} of \( Y \), characterized by
\[
\pi_* Y^H = Y.
\]

Let \( \nabla^B \) denote the Levi-Civita connection on \( TB \). We define a natural connection \( \nabla^H \) on \( T^H M \) by:
\[
\nabla^H_Y Z^H := (\nabla^B_Y Z)^H, \quad \text{for } Y, Z \in TB,
\]
\[
\nabla^H_Y Z^H := 0, \quad \text{for } Y \in TC, \ Z \in TB.
\]
Since \( T^H M \) inherits its scalar product from \( TB \), the connection \( \nabla^H \) is clearly metric with respect to this structure.

Let \( \nabla^M \) be the Levi-Civita connection on \( TM \), and let \( P_C \) (resp. \( P_H \)) denote the orthogonal projection onto \( TC \) (resp. \( T^H M \)). Then a connection \( \nabla^C \) on the vertical bundle \( TC \) is defined by:
\begin{equation} \label{verticalconn}
    \nabla^C_Y Z := P_C(\nabla^M_Y Z), \quad \text{for } Y \in TM, \ Z \in TC.
\end{equation}
The connection \( \nabla^C \) is compatible with the scalar product on \( TC \) by construction.

We now introduce a new connection on \( TM \), distinct from the Levi-Civita connection \( \nabla^M \). Let \( \bar{\nabla}^M \) denote the connection on \( TM \) that coincides with \( \nabla^H \) on the horizontal subbundle \( T^H M \), and with \( \nabla^C \) on the vertical subbundle \( TC \). The connection \( \bar{\nabla}^M \) is compatible with the Riemannian metric on \( TM \), and in particular, the splitting in~\eqref{split} is parallel with respect to \( \bar{\nabla}^M \). For a comparison between \( \bar{\nabla}^M \) and the Levi-Civita connection \( \nabla^M \), see~\cite{B86I}.

Now let \( X \) be a vertical Killing vector field on \( M \), and define the zero point set:
\[
M^X = \{ x \in M \mid X(x) = 0 \}.
\]
Assume that \( M^X \) is a fiber bundle over \( B \), with typical fiber denoted by \( C^X \). Let \( N^X \) be the normal bundle of \( M^X \) in \( M \), which is clearly a subbundle of \( TC \).

Define the infinitesimal action of \( X \) on \( N^X \) by:
\[
J_X : N^X \to N^X, \quad Y \mapsto \nabla_X Y.
\]

Let \( R \) denote the curvature tensor of \( TC \). For \( x \in M^X \) and \( Y, Z \in T_x M^X \), the normal bundle \( N^X \) is preserved under the action of the curvature \( R_x(Y, Z) \), and \( R_x(Y, Z) \) commutes with \( J_X \).

Under these assumptions, we have the following \emph{family version} of the Berline--Vergne~\cite{BV} and Duistermaat--Heckman~\cite{DH82} localization formula for integration along the fiber (see~\cite{B86} for details):
\begin{equation} \label{familylocal}
\int_{M/B} \omega = \int_{M^X/B} \frac{\omega}{\operatorname{Pf}\left[ \frac{J_X + R}{2\pi} \right]},
\end{equation}
where \( \omega \) is a differential form on \( M \) satisfying \( (d + i_X)\omega = 0 \), and \( \operatorname{Pf} \) denotes the Pfaffian of an antisymmetric matrix.

%--------------------------------------------------------------------------------------------------------------------

\subsection{Integration Along the Fiber in Loop Spaces} \label{loopbundle}

If $X$ is a finite dimensional smooth manifold, the \emph{free loop space} is
\[
  LX := C^\infty(S^1,X),
\] equipped with the Whitney $C^\infty$ topology (the
inverse limit of the $C^k$ topologies).  This makes $LX$ into a
\emph{Fréchet manifold}: for each $\gamma\in LX$ one chooses a
Riemannian metric on $X$ and uses the exponential map to identify a
neighbourhood of $\gamma$ with an open set in the Fréchet space
\[
  T_\gamma LX \;\cong\; \Gamma\!\big(S^1,\gamma^*TM\big),
\]
the space of smooth vector fields along $\gamma$.  

Looping the fiber bundle \( \pi: M \to B \), we obtain the induced fiber map on loop spaces:
\[
L\pi: LM \to LB.
\]
Let \( \{U_\alpha\} \) be a \emph{Brylinski open cover} of \( B \), i.e., a maximal open cover such that for any finite intersection \( U_{\alpha_I} := \bigcap_{i \in I} U_{\alpha_i} \), we have \( H^i(U_{\alpha_I}) = 0 \) for \( i = 2, 3 \), with \( |I| < \infty \) (see \cite{Bry}).

Assume that for each \( \alpha \), there exists a local trivialization
\[
\phi_{U_\alpha}: \pi^{-1}(U_\alpha) \to U_\alpha \times C,
\]
so that the bundle \( \pi: M \to B \) is locally trivial.

By applying the loop functor, we obtain a local trivialization of the looped bundle:
\[
L\phi_{U_\alpha}: \pi^{-1}(LU_\alpha) \to LU_\alpha \times LC,
\]
which endows \( LM \) with the structure of a smooth fiber bundle over \( LB \) with typical fiber \( LC \).

Let \( \gamma(t) \), \( 0 \leq t \leq 1 \), be a smooth loop in \( M \), i.e., a point in \( LM \). A tangent vector at \( \gamma \) is a smooth vector field \( X(\gamma(t)) \) along the loop \( \gamma \). If \( P_H[X(\gamma(t))] = X(\gamma(t)) \) for all \( t \), we call it a \emph{horizontal} tangent vector at \( \gamma \); similarly, if \( P_C[X(\gamma(t))] = X(\gamma(t)) \), we call it \emph{vertical}.

This defines a natural splitting of the tangent bundle:
\begin{equation} \label{loopsplit}
TLM = LT^H M \oplus LTC,
\end{equation}
where \( LT^H M \) and \( LTC \) denote the subbundles of horizontal and vertical tangent vectors in \( TLM \), respectively. This splitting induces a natural connection on the looped bundle \( L\pi: LM \to LB \).

On the tangent bundle \( TLB \), we define a natural scalar product \( g^{TLB} \) by
\[
g^{TLB}(X(\gamma(t)), Y(\gamma(t))) := \int_0^1 g^{TB}(X(\gamma(t)), Y(\gamma(t))) \, dt,
\]
for \( X(\gamma(t)), Y(\gamma(t)) \in TLB \). This scalar product on \( TLB \) induces, in turn, a scalar product \( g^{LT^H M} \) on the horizontal loop bundle \( LT^H M \). Specifically, for \( X(\gamma(t)), Y(\gamma(t)) \in LT^H M \), define
\[
g^{LT^H M}(X(\gamma(t)), Y(\gamma(t))) := \int_0^1 g^{TB}(\pi_*X(\gamma(t)), \pi_*Y(\gamma(t))) \, dt.
\]

Similarly, using the Riemannian metric \( g^{TC} \) on the vertical bundle \( TC \), one can define a scalar product \( g^{LTC} \) on \( LTC \). These scalar products on \( LT^H M \) and \( LTC \) together define a scalar product \( g^{TLM} \) on \( TLM \), by declaring the two subbundles to be orthogonal:
\[
TLM = LT^H M \oplus LTC.
\]

The bundles with connections
\[
(TB, \nabla^B), \quad (T^H M, \nabla^H), \quad (TC, \nabla^C), \quad (TM, \nabla^M), \quad \text{and } (TM, \bar{\nabla}^M)
\]
naturally induce the looped bundles with their respective looped connections:
\[
(TLB, \nabla^{LB}), \quad (LT^H M, \nabla^{LH}), \quad (LTC, \nabla^{LC}), \quad (TLM, \nabla^{LM}), \quad \text{and } (TLM, \bar{\nabla}^{LM}).
\]

It is straightforward to verify that \( \bar{\nabla}^{LM} \) preserves the scalar product \( g^{TLM} \), and that the splitting in~\eqref{loopsplit} is parallel with respect to \( \bar{\nabla}^{LM} \).

$\, $

Let \( K \) denote the vector field on \( LM \) that generates the natural rotation action on loops (i.e., the infinitesimal generator of loop rotation). Then the projection \( P_C(K) \) defines a \emph{vertical vector field} along loops, taking values in the vertical bundle \( LTC \).

Let \( \omega \in \Omega^*(LM) \) be a differential form satisfying the equivariant closedness condition:
\[
(d + i_K)\omega = 0.
\]
Fix a loop \( \gamma \in LB \), and restrict attention to the fiber \( L\pi^{-1}(\gamma) \subset LM \). Denote by \( i \) the inclusion map of the fiber into \( LM \). Then we have:
\[
(d + i_{P_C(K)})\, i^* \omega = 0.
\]
This naturally leads one to consider a \emph{family localization formula} analogous to~\eqref{familylocal}.

However, two major difficulties arise:
\begin{itemize}
  \item[(a)] The zero set of the vertical vector field \( P_C(K) \),
  \[
  LM^{P_C(K)} := \{ \gamma \in LM \mid P_C(K)_\gamma = 0 \},
  \]
  may fail to form a smooth fiber bundle over \( LB \).
  
  \item[(b)] The fibers of \( LM \to LB \) are infinite-dimensional, which presents analytic and geometric challenges not encountered in the finite-dimensional setting.
\end{itemize}

$\, $

\emph{In the following, we assume that the fiber bundle \( \pi: M \to B \), equipped with the connection from~\eqref{split}, is flat}. This means that the curvature tensor
\[
\Omega(X, Y) = P_C([X^H, Y^H]), \quad \text{for } X, Y \in TM,
\]
vanishes identically.

For a loop \( \gamma \in LM \), and any point \( \gamma(t) \) in the loop, one may parallel transport back along the projected path \( \pi \circ \gamma(s) \), for \( 0 \leq s \leq t \), to obtain a point \( P\gamma(t) \) in the fiber \( C_{\pi \circ \gamma(0)} \). Since the bundle is flat, the entire loop \( \gamma(t),\, 0 \leq t \leq 1 \) corresponds to a loop \( P\gamma(t) \) lying entirely in the single fiber \( C_{\pi \circ \gamma(0)} \). Moreover, the vertical component of the rotation vector field, \( P_C(K_{\gamma(t)}) \), is parallel transported to the corresponding rotation vector \( K_{P\gamma(t)} \) along the image loop.

Now consider the fiber bundle \( \widetilde{\pi}: \widetilde{M} \to LB \), whose fiber over a loop \( \beta(t) \in LB \) is the fiber \( C_{\beta(0)} \). This bundle has typical fiber \( C \), and inherits geometric data such as connections and metrics from the loop bundle \( L\pi: LM \to LB \). In what follows, we will use \( (T\widetilde{C}, \nabla^{\widetilde{C}}) \) to denote the vertical tangent bundle and the vertical connection on \( \widetilde{M} \), in order to distinguish it from \( (TC, \nabla^C) \) defined in~\eqref{verticalconn}.

From the above discussion, one observes that the loop bundle \( L\pi: LM \to LB \) can be identified with the \emph{vertical loop space} of the bundle \( \widetilde{\pi}: \widetilde{M} \to LB \), i.e.,
\[
L^{\mathbf{v}}\widetilde{\pi}: L^{\mathbf{v}} \widetilde{M} \to LB.
\]
Let \( P: LM \to L^{\mathbf{v}} \widetilde{M} \) denote the natural bundle isomorphism.

For a loop \( \gamma \in LM^{P_C(K)} \), i.e., a loop where \( K \) is entirely horizontal, it is easy to see that the image \( P\gamma \) is simply the constant loop at \( \gamma(0) \in LC_{\pi \circ \gamma(0)} \). Hence, the inclusion \( LM^{P_C(K)} \hookrightarrow LM \) corresponds under the map \( P \) to the inclusion \( \widetilde{M} \hookrightarrow L^{\mathbf{v}} \widetilde{M} \).

Note that \( L^{\mathbf{v}} \widetilde{M} \) carries a natural fiberwise circle action, whose fixed-point set is precisely the bundle 
\( \widetilde{M} \subset L^{\mathbf{v}} \widetilde{M} \).

Any differential form \( \omega \in \Omega^*(LM) \) naturally induces a form \( \omega^0 \in \Omega^*(\widetilde{M}) \). Indeed, for any vertical vector \( X_0 \) at a point \( p \in \widetilde{M} \), one may parallel transport \( p \) to a loop \( \gamma(t) \in LM^{P_C(K)} \), and \( X_0 \) to a time-dependent vector field \( X_t \) along \( \gamma(t) \). We then define
\[
\omega^0(X_0) := \omega(X(t)).
\]

$\,$

We are now in a position to apply the family localization formula~\eqref{familylocal}. Let \( \omega \in \Omega^*(LM) \) be a differential form satisfying the equivariant closure condition:
\[
(d + i_{P_C(K)}) \omega = 0.
\]
Then its pullback via the identification map \( P: LM \to L^{\mathbf{v}} \widetilde{M} \) satisfies:
\[
(d + i_{K_{P\gamma}}) \left( (P^{-1})^* \omega \right) = 0.
\]

By formally applying the family localization formula~\eqref{familylocal} to the vertical loop bundle \( L^{\mathbf{v}} \widetilde{M} \to LB \), and invoking the ideas from Atiyah--Witten~\cite{A85} and Bismut~\cite{B85} on integration and localization in loop space, we arrive at the following \emph{formal} formula (with the last integral rigorously defined, since the fiber is finite-dimensional):
\begin{equation} \label{loopintegration}
\int_{LM/LB} \omega = \int_{L^{\mathbf{v}} \widetilde{M} / LB} (P^{-1})^* \omega = \int_{\widetilde{M}/LB} \omega^0 \cdot \widehat{A}(T\widetilde{C}, \nabla^{\widetilde{C}}).
\end{equation}

Here, the differential form \( \widehat{A}(T\widetilde{C}, \nabla^{\widetilde{C}}) \) is the \emph{Hirzebruch \( \widehat{A} \)-form}, given explicitly by:
\begin{equation}
\widehat{A}(T\widetilde{C}, \nabla^{\widetilde{C}}) := \det{}^{1/2} \left( \frac{ \frac{\sqrt{-1}}{4\pi} (\nabla^{\widetilde{C}})^2 }{ \sinh\left( \frac{\sqrt{-1}}{4\pi} (\nabla^{\widetilde{C}})^2 \right) } \right),
\end{equation}
which appears, for example, in Section~1.6.3 of~\cite{Z01}.

%--------------------------------------------------------------------------------------------------------------------
\subsection{The Case of Principal Circle Bundles} \label{circle bundle}

Now let
\[
\begin{CD}
\mathbb{T} @>>> M \\
@. @V{\pi}VV \\
@. B
\end{CD}
\]
be a principal circle bundle equipped with a connection and the necessary geometric structures as described earlier. Let \( v \) denote the vector field generating the \( \mathbb{T} \)-action on \( M \).

Looping this bundle as in Section~\ref{loopbundle}, we obtain a principal \( L\mathbb{T} \)-bundle:
\[
\begin{CD}
L\mathbb{T} @>>> LM \\
@. @V{L\pi}VV \\
@. LB
\end{CD}
\]
The vector field \( v \) induces a looped vector field \( v_L \) on \( LM \), defined pointwise by \( v_L|_{\gamma(t)} = v|_{\gamma(t)} \).

We also consider the principal bundle
\[
\begin{CD}
\mathbb{T} @>>> \widetilde{M} \\
@. @V{\widetilde{\pi}}VV \\
@. LB
\end{CD}
\]
on which \( v \) similarly induces a vector field \( \widetilde{v} \) on \( \widetilde{M} \).

Suppose the original bundle \( \pi: M \to B \) is flat. Take a differential form \( \omega \in \Omega^*(LM) \) satisfying the equivariant closedness condition \( (d + i_{P_C(K)})\omega = 0 \). Then by the formal localization formula~\eqref{loopintegration}, one has:
\[
\int_{LM/LB} \omega = \int_{\widetilde{M}/LB} \omega^0 \cdot \widehat{A}(T\widetilde{\mathbb{T}}, \nabla^{\widetilde{\mathbb{T}}}).
\]
In this case, it is straightforward to see that the vertical tangent bundle \( T\widetilde{\mathbb{T}} \) is trivial and its curvature vanishes, i.e., \( (\nabla^{\widetilde{\mathbb{T}}})^2 = 0 \). Therefore,
\[
\widehat{A}(T\widetilde{\mathbb{T}}, \nabla^{\widetilde{\mathbb{T}}}) = 1.
\]
As a result, the formal formula simplifies to:
\begin{equation} \label{loopcirclebundle}
\int_{LM/LB} \omega = \int_{\widetilde{M}/LB} \omega^0.
\end{equation}

Thus, for the loop bundle of a flat principal circle bundle, we obtain a formal definition of fiber integration via~\eqref{loopcirclebundle}.

{\em Of course, in T-duality, the principal circle bundles in question are typically not flat.} {\em Nonetheless, motivated by the formal identity~\eqref{loopcirclebundle}, we will give a rigorous definition of \( \int_{LM/LB} \omega \) for \( \omega \in \Omega^*(LM)^{L\mathbb{T}} \), i.e., for \( L\mathbb{T} \)-invariant forms on \( LM \), in the following.}

$\, $

Let \( \omega \in \Omega^*(LM)^{L\mathbb{T}} \). Since the space of \( L\mathbb{T} \)-invariant vector fields on \( LM \) forms a rank-one module over \( C^\infty(LM) \), the contraction \( i_{v_L} \omega \) defines a well-defined differential form on \( LB \).

\begin{definition} \label{intfiberloop}
Let \( \omega \in \Omega^*(LM)^{L\mathbb{T}} \). Define the \emph{integration along the fiber} as:
\[
\int_{LM/LB} \omega := i_{v_L} \omega.
\]
\end{definition}

\begin{remark}
\begin{enumerate}[(a)]
    \item The condition \( (d + i_{P_C(K)}) \omega = 0 \) is not required in this definition. We only assume that \( \omega \) is \( L\mathbb{T} \)-invariant.
    
    \item When the circle bundle is flat, this definition agrees with the formal identity~\eqref{loopcirclebundle}:
    \[
    \int_{LM/LB} \omega = i_{\widetilde{v}} \omega^0 = \int_{\widetilde{M}/LB} \omega^0.
    \]
    
    \item This definition can generalize naturally to principal \( G \)-bundles for any Lie group \( G \), using contraction with respect to the looped volume form.
\end{enumerate}
\end{remark}

We now record some basic properties of the fiber integration operation.

\begin{theorem}
For \( \omega \in \Omega^*(LM)^{L\mathbb{T}} \), one has:
\begin{align}
\int_{LM/LB} d\omega &= - d \int_{LM/LB} \omega, \label{prop-d} \\
\int_{LM/LB} i_K \omega &= - i_K \int_{LM/LB} \omega. \label{prop-i}
\end{align}
\end{theorem}

\begin{remark}
In~\eqref{prop-i}, we slightly abuse notation: on the left-hand side, \( K \) denotes the rotation vector field on \( LM \), while on the right-hand side, it refers to the induced rotation vector field on \( LB \).
\end{remark}

\begin{proof}
Since \( \omega \) is \( L\mathbb{T} \)-invariant, we compute:
\[
\int_{LM/LB} d\omega = i_{v_L} d\omega = (L_{v_L} - d\, i_{v_L}) \omega = - d(i_{v_L} \omega) = - d \int_{LM/LB} \omega.
\]

To prove~\eqref{prop-i}, we use \( [v_L, K] = 0 \), hence
\[
\int_{LM/LB} i_K \omega = i_{v_L} i_K \omega = (i_{[v_L, K]} - i_K i_{v_L}) \omega = - i_K i_{v_L} \omega = - i_K \int_{LM/LB} \omega.
\]
The desired results follow.
\end{proof}

%%%%%%%%%%%%%%%%%%%%%%%%%%%%%%%%%%%%%%%%%%%%%%%%%%%%%%
\section{The Twisted Bismut-Chern Character of the Poincar\'e Gerbe Module on the Correspondence Space} \label{twistedBCh}

In this section, we begin by reviewing the theory of exotic twisted equivariant cohomology on loop spaces and the twisted Bismut--Chern character, as developed in~\cite{HM15}, in Section~\ref{exotic}. Subsequently, in Section~\ref{reviewT}, we revisit the framework of T-duality with \( H \)-flux, following the foundational works~\cite{BEM04a, BEM04b}.

With these preparations in place, we proceed in Section~\ref{Poincare} to apply the loop space theory to the \emph{Poincar\'e gerbes module} associated with the spaces appearing in T-duality. This leads to the construction of the \emph{horizontal twisted Bismut--Chern character}, which plays a central role in formulating the loop Hori formula in Section~\ref{LoopHori}.

%--------------------------------------------------------------------------------------------------------------------
\subsection{Review of Exotic Twisted Equivariant Cohomology of Loop Space and Twisted Bismut-Chern Character} \label{exotic}

\subsubsection{Exotic Twisted Equivariant Cohomology of Loop Space} \label{reviewExotic}

Let \( M \) be a smooth manifold, and let \( \{\sU_\alpha\} \) be a Brylinski open cover of \( M \). To distinguish this general setting from the special case of principal circle bundles to be discussed later in the context of T-duality, we use serif font type to denote open subsets of \( M \).

Define the \emph{transgression map}
\begin{equation} \label{eqn:trans}
\tau: \Omega^\bullet(\sU_{\alpha_I}) \longrightarrow \Omega^{\bullet-1}(L\sU_{\alpha_I}),
\end{equation}
by
\begin{equation}
\tau(\xi_I) = \int_{S^1} \mathrm{ev}^*(\xi_I), \qquad \xi_I \in \Omega^\bullet(\sU_{\alpha_I}),
\end{equation}
where
\[
\mathrm{ev} : S^1 \times LM \to M, \quad (t, x) \mapsto x(t),
\]
is the evaluation map.

Let \( \omega \in \Omega^i(M) \). For each \( s \in [0, 1] \), define a differential form \( \widetilde{\omega}_s \in \Omega^i(LM) \) by:
\[
\widetilde{\omega}_s(X_1, \ldots, X_i)(x) := \omega\left(X_1|_{x(s)}, \ldots, X_i|_{x(s)}\right),
\]
where \( x \in LM \) and \( X_1, \ldots, X_i \) are vector fields on \( LM \) defined in a neighborhood of \( x \). One can verify that this construction satisfies
\[
d \widetilde{\omega}_s = \widetilde{d\omega}_s.
\]
Such constructions on loop space were previously studied in~\cite{B85}.

Now consider the averaged \( i \)-form on the loop space:
\[
\overline{\omega} := \int_0^1 \widetilde{\omega}_s \, ds \in \Omega^i(LM).
\]
This averaged form is \( S^1 \)-invariant; that is,
\[
L_K(\overline{\omega}) = 0,
\]
where \( K \) denotes the vector field on \( LM \) generating loop rotation. Moreover, it satisfies
\[
d\overline{\omega} = \overline{d\omega}, \qquad \tau(\omega) = i_K \overline{\omega}.
\]
We refer to \( \overline{\omega} \) as the \emph{average} of \( \omega \).

Let \( (H, B_\alpha, F_{\alpha\beta}, (L_{\alpha\beta}, \nabla^{L_{\alpha\beta}})) \) be a gerbe with connection on \( M \), where:
\begin{itemize}
  \item \( (H, B_\alpha, F_{\alpha\beta}) \) denotes the Deligne class of a closed integral 3-form \( H \),
  \item \( (L_{\alpha\beta}, \nabla^{L_{\alpha\beta}}) \) denotes a line bundle with connection on double overlaps \( \sU_\alpha \cap \sU_\beta \) defining the gerbe.
\end{itemize}

On each triple intersection \( \sU_\alpha \cap \sU_\beta \cap \sU_\gamma \), there exists a trivialization:
\begin{equation} \label{triv}
(L_{\alpha\beta}, \nabla^{L_{\alpha\beta}}) \otimes (L_{\beta\gamma}, \nabla^{L_{\beta\gamma}}) \otimes (L_{\gamma\alpha}, \nabla^{L_{\gamma\alpha}}) \simeq (\mathbb{C}, d).
\end{equation}

For any loop \( x \in L\sU_\alpha \cap L\sU_\beta \), i.e., \( x: S^1 \to \sU_\alpha \cap \sU_\beta \), consider the parallel transport equation in the line bundle with connection \( (L_{\alpha\beta}, \nabla^{L_{\alpha\beta}}) \):
\begin{equation} \label{ptran}
\nabla_{\dot{x}(s)} \tau^0_s = 0, \quad \tau^0_0 = \mathrm{Id},
\end{equation}
where \( \tau^0_s \in \mathrm{Hom}(L_{\alpha\beta}|_{x(0)}, L_{\alpha\beta}|_{x(s)}) \), and \( \dot{x}(s) \) denotes the tangent vector to the loop at time \( s \).

The holonomy of this parallel transport defines a smooth function
\begin{equation} \label{transition}
g_{\alpha\beta} = \mathrm{Tr}(\tau^0_1) = \mathrm{hol}(\nabla^{L_{\alpha\beta}})
\end{equation}
on \( L\sU_\alpha \cap L\sU_\beta \).

From the trivialization~\eqref{triv}, one sees that the holonomies satisfy the cocycle condition:
\[
g_{\alpha\beta} \cdot g_{\beta\gamma} \cdot g_{\gamma\alpha} = 1.
\]

The holonomy line bundle of the gerbe is defined by:
\begin{equation} \label{holbundle}
\mathcal{L}^B := \left( \coprod_{\alpha \in I} \{ \alpha \} \times L\sU_\alpha \times \mathbb{C} \right) \bigg/ \sim,
\end{equation}
where the equivalence relation is given by
\[
(\beta, x, w) \sim (\alpha, x, g_{\alpha\beta}(x) \cdot w), \quad \forall\, \alpha, \beta \in I,\ x \in L(\sU_\alpha \cap \sU_\beta),\ w \in \mathbb{C}.
\]
We denote by \( \sigma_\alpha = (x, 1) \) the local section of \( \mathcal{L}^B \) over \( L\sU_\alpha \).

\begin{proposition}[\textbf{Brylinski}, Section~6.1 in~\cite{Bry}] \label{gauge}
The system of 1-forms \( \{ -i_K \overline{B_\alpha} \} \) satisfies the gauge transformation law:
\[
d \ln g_{\alpha\beta}^{-1} = -i_K \overline{B_\alpha} + i_K \overline{B_\beta}.
\]
\end{proposition}

The previous proposition enables us to define a natural connection on the holonomy line bundle \( \mathcal{L}^B \). Specifically, let \( \nabla^{\mathcal{L}^B} \) be the connection on \( \mathcal{L}^B \) such that, over \( L\sU_\alpha \), under the trivialization given by the local section \( \sigma_\alpha \), the connection one-form is \( -i_K \overline{B_\alpha} \).

Consider the space \( \Omega^\bullet(LM, \mathcal{L}^B) \) of differential forms on the loop space \( LM \) with values in the holonomy line bundle \( \mathcal{L}^B \to LM \). Define the following superconnection:
\[
D_{\overline{H}} := \nabla^{\mathcal{L}^B} - i_K + \overline{H}.
\]
In~\cite{HM15}, it was shown that this superconnection is \( S^1 \)-equivariantly flat.

\begin{theorem}[\cite{HM15}] \label{flat}
The superconnection \( D_{\overline{H}} \) satisfies
\[
(D_{\overline{H}})^2 = 0
\]
on \( \Omega^\bullet(LM, \mathcal{L}^B)^{S^1} \).
\end{theorem}

This \( S^1 \)-equivariant flatness allows one to define a cohomology theory on the loop space. The \textbf{completed periodic exotic twisted \( S^1 \)-equivariant cohomology},
\[
h^\bullet_{S^1}(LM, \mathcal{L}^B; \overline{H}),
\]
introduced in~\cite{HM15}, is defined as the cohomology of the complex
\[
\left( \Omega^\bullet(LM, \mathcal{L}^B)^{S^1}[[u, u^{-1}]],\ \nabla^{\mathcal{L}^B} - u i_K + u^{-1} \overline{H} \right),
\]
where \( \deg(u) = 2 \).

In~\cite{HM15}, the following localization theorem was established. It localizes the completed periodic exotic twisted \( S^1 \)-equivariant cohomology to the twisted de Rham cohomology of the fixed point submanifold of \( LM \), which is canonically identified with \( M \).

\begin{theorem}[\cite{HM15}] \label{local}
There is a natural isomorphism:
\[
h^\bullet_{S^1}(LM, \nabla^{\mathcal{L}^B}; \overline{H}) \cong H^\bullet\left( \Omega(M)[[u, u^{-1}]], d + u^{-1} H \right) \cong H^\bullet(M, H)[[u, u^{-1}]].
\]
\end{theorem}

$\, $

\subsubsection{Twisted Bismut-Chern Character of Gerbe Modules}

In~\cite{BM00}, it was proposed that D-brane charges in the presence of a background \( H \)-flux take values in the twisted \( K \)-theory of spacetime \( M \), denoted \( K^\bullet(M, H) \). The Chern--Weil representatives of the twisted Chern character
\[
Ch_H: K^\bullet(M, H) \longrightarrow H^\bullet(M, H)
\]
were constructed and studied in~\cite{BCMMS}.

In~\cite{HM15}, it was further shown that the twisted Chern character \( Ch_H \) admits a refinement to a \emph{twisted Bismut--Chern character}
\[
BCh_H: K^\bullet(M, H) \longrightarrow h^\bullet_{S^1}(LM, \mathcal{L}^B; \overline{H}),
\]
which satisfies the following commutative diagram:
\begin{equation*}
\xymatrix@=3pc{
K^\bullet(M, H) \ar[dr]_{Ch_H} \ar[rr]^{BCh_H} && h^\bullet_{S^1}(LM, \mathcal{L}^B; \overline{H}) \ar[dl]^{\mathrm{res}} \\
& H^\bullet(\Omega(M)[[u, u^{-1}]], d + u^{-1} H) &
}
\end{equation*}

\smallskip

When \( H = 0 \), this construction reduces to the original Bismut--Chern character introduced in~\cite{B85}.

We begin by recalling the geometric representatives of twisted \( K \)-theory and the Chern--Weil representative of the twisted Chern character.

Let \( E = \{E_\alpha\} \) be a collection of (infinite-dimensional) Hilbert bundles \( E_\alpha \to \sU_\alpha \), each with structure group reduced to \( U_{\mathfrak{I}} \), the group of unitary operators on a fixed model Hilbert space \( \mathcal{H} \) of the form identity plus a trace-class operator. Here, \( \mathfrak{I} \) denotes the Lie algebra of trace-class operators on \( \mathcal{H} \).

Assume further that on double overlaps \( \sU_{\alpha\beta} \), there exist bundle isomorphisms
\begin{equation} \label{gerbeconn}
\phi_{\alpha\beta}: L_{\alpha\beta} \otimes E_\beta \xrightarrow{\cong} E_\alpha,
\end{equation}
which are compatible on triple overlaps due to the gerbe cocycle condition. In this case, the collection \( \{E_\alpha\} \) is called a \emph{gerbe module} for the gerbe \( \{L_{\alpha\beta}\} \).

A \emph{gerbe module connection} \( \nabla^E \) is a collection of connections \( \{\nabla^E_\alpha\} \), where each
\[
\nabla^E_\alpha = d + A^E_\alpha, \quad A^E_\alpha \in \Omega^1(\sU_\alpha) \otimes \mathfrak{I},
\]
and the curvatures \( F^E_\alpha \) satisfy the compatibility condition on overlaps:
\[
\phi_{\alpha\beta}^{-1} (F^E_\alpha) \phi_{\alpha\beta} = F^L_{\alpha\beta} \cdot I + F^E_\beta.
\]
Using the relation~\eqref{gerbeconn}, this condition can be rewritten as
\[
\phi_{\alpha\beta}^{-1} (B_\alpha \cdot I + F^E_\alpha) \phi_{\alpha\beta} = B_\beta \cdot I + F^E_\beta.
\]
It follows that the differential form
\[
\exp(-B) \cdot \mathrm{Tr}\left( \exp(-F^E) - I \right)
\]
is globally defined on \( M \) and has even total degree. Note that \( \mathrm{Tr}(I) = \infty \), so the subtraction of the identity is necessary to ensure convergence.

%%%%

Let $E=\{E_{\alpha}\}$ and $E'=\{E'_{\alpha}\}$ be
be a {gerbe modules} for the gerbe $\{L_{\alpha\beta}\}$. Then an element of twisted K-theory $K^0(M, H)$
is represented by the pair $(E, E')$, see \cite{BCMMS}. Two such pairs $(E, E')$ and $(G, G')$ are equivalent
if $E\oplus G' \oplus K \cong E' \oplus G \oplus K$ as gerbe modules for some gerbe module $K$ for the gerbe $\{L_{\alpha\beta}\}$.
We can assume without loss of generality that these gerbe modules $E, E'$ are modeled on the same Hilbert space
$\gH$, after a choice of isomorphism if necessary.

Suppose that $\nabla^E, \nabla^{E'}$ are gerbe module connections on the gerbe modules $E, E'$ respectively. Then we can define the {\bf twisted Chern character} as
\begin{align*}
Ch_H &: K^0(M, H) \to H^{even}(M, H)\\
Ch_H(E, E')&= \exp(-B)\Tr\left(\exp(-F^E) - \exp(-F^{E'})\right)
\end{align*}
That this is a well defined homomorphism is explained in \cite{BCMMS,MS}. To define the twisted Chern character landing in $\left(\Omega^\bu(M)[[u, u^{-1}]]\right)_{(d+u^{-1}H)-cl}$, simply replace the above formula by
$$Ch_H(E, E')= \exp(-u^{-1}B)\Tr\left(\exp(-u^{-1}F^E) - \exp(-u^{-1}F^{E'})\right).$$

Let \( \xi \) be a bundle with connection. Denote by \( {}^\xi\tau^0_s \) the parallel transport along a loop \( \gamma \) from \( \gamma_0 \) to \( \gamma_s \), and let \( {}^\xi\tau^s_0 := ({}^\xi\tau^0_s)^{-1} \).

We define a local twisted Bismut--Chern character form
\[
BCh_{H, \alpha}(\nabla^E, \nabla^{E'}) \in \Omega^\bullet(L\sU_\alpha, \mathcal{L}^B)^{S^1}[[u, u^{-1}]]
\]
by the expression:
{\footnotesize
\begin{equation} \label{localBCh}
\begin{aligned}
&BCh_{H, \alpha}(\nabla^E, \nabla^{E'}) \\
 &= \left(1 + \sum_{n=1}^\infty (-u)^{-n} \int_{0 \leq s_1 \leq \cdots \leq s_n \leq 1} \widetilde{B_\alpha}_{s_1} \cdots \widetilde{B_\alpha}_{s_n} \right) \\
& \cdot \mathrm{Tr}\Bigg[ \Bigg(I + \sum_{n=1}^\infty (-u)^{-n} \int_{0 \leq s_1 \leq \cdots \leq s_n \leq 1} 
{}^{E_\alpha}\tau^0_{s_1}(\widetilde{F^E_\alpha}_{s_1}) \circ \cdots \circ {}^{E_\alpha}\tau^0_{s_n}(\widetilde{F^E_\alpha}_{s_n}) \Bigg) \circ {}^{E_\alpha}\tau^1_0 \\
&\quad - \Bigg(I + \sum_{n=1}^\infty (-u)^{-n} \int_{0 \leq s_1 \leq \cdots \leq s_n \leq 1} 
{}^{E'_\alpha}\tau^0_{s_1}(\widetilde{F^{E'}_\alpha}_{s_1}) \circ \cdots \circ {}^{E'_\alpha}\tau^0_{s_n}(\widetilde{F^{E'}_\alpha}_{s_n}) \Bigg) \circ {}^{E'_\alpha}\tau^1_0 \Bigg] \otimes \sigma_\alpha \\
=\
& e^{-u \overline{B_\alpha}} \cdot \mathrm{Tr}\Bigg[ \Bigg(I + \sum_{n=1}^\infty (-u)^{-n} \int_{0 \leq s_1 \leq \cdots \leq s_n \leq 1} 
{}^{E_\alpha}\tau^0_{s_1}(\widetilde{F^E_\alpha}_{s_1}) \circ \cdots \circ {}^{E_\alpha}\tau^0_{s_n}(\widetilde{F^E_\alpha}_{s_n}) \Bigg) \circ {}^{E_\alpha}\tau^1_0 \\
&\quad - \Bigg(I + \sum_{n=1}^\infty (-u)^{-n} \int_{0 \leq s_1 \leq \cdots \leq s_n \leq 1} 
{}^{E'_\alpha}\tau^0_{s_1}(\widetilde{F^{E'}_\alpha}_{s_1}) \circ \cdots \circ {}^{E'_\alpha}\tau^0_{s_n}(\widetilde{F^{E'}_\alpha}_{s_n}) \Bigg) \circ {}^{E'_\alpha}\tau^1_0 \Bigg] \otimes \sigma_\alpha
\end{aligned}
\end{equation}
}

\noindent
In the above expression, the volume form \( ds_1 \cdots ds_n \) in the integrals is omitted for brevity. The subtraction of terms involving \( E \) and \( E' \) is crucial to ensure that the trace is well-defined.

The local differential forms \( \{ BCh_{H, \alpha} \} \) patch together to form a global differential form
\[
BCh_H(\nabla^E, \nabla^{E'}) \in \Omega^\bullet(LM, \mathcal{L}^B)^{S^1}[[u, u^{-1}]],
\]
called the \textbf{twisted Bismut--Chern character form}.

\begin{theorem} \label{BCh prop}
The following properties hold:
\begin{enumerate}[(i)]
\item \( (\nabla^{\mathcal{L}^B} - u i_K + u^{-1} \overline{H}) BCh_H(\nabla^E, \nabla^{E'}) = 0 \);
\item The exotic twisted \( S^1 \)-equivariant cohomology class \( [BCh_H(\nabla^E, \nabla^{E'})] \) is independent of the choice of connections \( \nabla^E \), \( \nabla^{E'} \).
\end{enumerate}
\end{theorem}

%--------------------------------------------------------------------------------------------------------------------
\subsection{T-duality with H-flux} \label{reviewT}

In \cite{BEM04a, BEM04b}, spacetime $Z$ was compactified in one direction.
More precisely, $Z$
is a principal $\bbT$-bundle over $X$

\begin{equation}\label{eqn:MVBx}
\begin{CD}
\bbT @>>> Z \\
&& @V\pi VV \\
&& X \end{CD}
\end{equation}
 classified up to isomorphism by its first Chern class
{ $c_1(Z)\in H^2(X,\ZZ)$}. Assume that spacetime $Z$ is endowed with an $H$-flux which is
a representative in the
degree 3 Deligne cohomology of $Z$, that is
$H\in\Omega^3(Z)$ with integral periods ({\em for simplicity, here we drop factors of $\frac{1}{2\pi i}$ and in the following such normalization numbers are all dropped}),
together with
the following data. Consider a local trivialization $U_\alpha \times \TT$ of $Z\to X$, where
$\{U_\alpha\}$ is a good cover of $X$. Let $H_\alpha = H\Big|_{ U_\alpha \times \TT}
= d B_\alpha$, where $B_\alpha \in \Omega^2(U_\alpha \times \TT)$ and finally, 
\be \label{Bfield} 
B_\alpha -B_\beta = F_{\alpha\beta}
\in \Omega^1(U_{\alpha\beta} \times \TT).
\ee
 Then the choice of $H$-flux entails that we are given a local trivialization
 as above and locally defined 2-forms $B_\alpha$ on it, together with closed 2-forms $F_{\alpha\beta}$ defined on double overlaps,  that is, $(H, B_\alpha, F_{\alpha\beta})$. Also the first Chern class
 of $Z\to X$
  is represented in integral cohomology by $(F, A_\alpha)$ where
$\{A_\alpha\}$ is a connection 1-form on $Z\to X$ and $F = dA_\alpha$ is the curvature 2-form of $\{A_\alpha\}$.

The {  {T-dual}}  is another principal
$\bbT$-bundle over $M$, denoted by $\widehat Z$,
  {}
\begin{equation}\label{eqn:MVBy}
\begin{CD}
\widehat \bbT @>>> \widehat Z \\
&& @V\widehat \pi VV     \\
&& X \end{CD}
\end{equation}
To define it, we see that $\pi_* (H_\alpha) = d \pi_*(B_\alpha) = d {\widehat A}_\alpha$,
 so that $\{{\widehat A}_\alpha\}$ is a connection 1-form whose curvature $ d {\widehat A}_\alpha = \widehat F_\alpha =  \pi_*(H_\alpha)$
 that is, $\widehat F = \pi_* H$. So let $\widehat Z$ denote the principal
$\bbT$-bundle over $M$ whose first Chern class is  $\,\, c_1(\widehat Z) = [\pi_* H, \pi_*(B_\alpha)] \in H^2(X; \ZZ) $.

The Gysin
sequence for $Z$ enables us to define a T-dual $H$-flux
$[\widehat H]\in H^3(\widehat Z,\ZZ)$, satisfying
\begin{equation} \label{eqn:MVBc}
c_1(Z) = \widehat \pi_* \widehat H \,,
\end{equation}
 where $\pi_* $
and similarly $\widehat\pi_*$, denote the pushforward maps.
Note that $ \widehat H$ is not fixed by this data, since any integer
degree 3 cohomology class on $X$ that is pulled back to $\widehat Z$
also satisfies the requirements. However, $ \widehat H$ is
determined uniquely (up to coboundary) upon imposing
the condition $[H]=[\widehat H]$ on the correspondence space $Z\times_X \widehat Z$
as will be explained now.

The {\em correspondence space} (sometimes called the doubled space) is defined as
$$
Z\times_X  \widehat Z = \{(x, \widehat x) \in Z \times \widehat Z: \pi(x)=\widehat\pi(\widehat x)\}.
$$
Then we have the following commutative diagram,
\begin{equation} \label{eqn:correspondence}
\xymatrix @=6pc @ur { (Z, [H]) \ar[d]_{\pi} &
(Z\times_X  \widehat Z, [H]=[\widehat H]) \ar[d]_{\widehat p} \ar[l]^{p} \\ X & (\widehat Z, [\widehat H])\ar[l]^{\widehat \pi}}
\end{equation}
By requiring that
$$
p^*[H]={\widehat p}^*[\widehat H] \in H^3(Z\times_X  \widehat Z, \ZZ),
$$
determines $[\widehat H] \in H^3(  \widehat Z, \ZZ)$  uniquely, via an application of the Gysin sequence.

An alternate way to see this is is explained below.

%\subsection{Formula for the T-dual flux in Deligne cohomology}\label{subsect:T-dualitydeligne}
Let $(H, B_\alpha, F_{\alpha\beta}, (L_{\alpha\beta}, \nabla^{L_{\alpha\beta}}))$  denote a gerbe with connection on $Z$.
We also choose a connection 1-form $A$ on $Z$.
Let $v$ denote the vectorfield generating the $S^1$-action on $Z$.
Then define 
\be \label{newconn} \widehat A_\alpha = -i_v B_\alpha\ee 
on the chart $U_\alpha$ and
the connection 1-form $\widehat A= \widehat A_\alpha +d\widehat\theta_\alpha$
on the chart $U_\alpha\times  \widehat \TT$, where $\theta_\alpha$ is the coordinate on $\widehat \TT$. 
In this way we get a T-dual circle bundle
$\widehat Z \to X$ with connection 1-form $\widehat A$.

Without loss of generality, we can assume that $H$ is $\TT$-invariant. Consider
$$
\Omega = H - A\wedge F_{\widehat A}
$$
where  $F_{\widehat A} = d {\widehat A}$ and $F_{A} = d {A}$ are the curvatures of $A$
and $\widehat A$ respectively. One checks that the contraction $i_v(\Omega)=0$ and
the Lie derivative $L_v(\Omega)=0$ so that $\Omega$ is a basic 3-form on $Z$, that is
$\Omega$ comes from the base $X$.

Setting
$$
\widehat H = F_A\wedge {\widehat A} + \Omega
$$
this defines the T-dual flux 3-form. One verifies that $\widehat H$ is a closed 3-form on $\widehat Z$.
It follows that on the correspondence space, one has as desired,
\begin{equation}
\widehat H = H + d (A\wedge \widehat A ).
\end{equation}

Our next goal is to determine the T-dual curving or B-field.
The Buscher rules imply that on the open sets $U_\alpha \times \TT\times \widehat \TT$ of the
correspondence space $Z\times_X \widehat Z$, one has
\begin{equation} \label{buscher}
\widehat B_\alpha = B_\alpha + A\wedge \widehat A - d\theta_\alpha \wedge d\widehat \theta _\alpha\,,
\end{equation}
Note that
\begin{equation}
i_v \widehat B_\alpha = i_v
\left( B_\alpha + A\wedge \widehat A - d\theta_\alpha \wedge d\widehat \theta _\alpha\right) =
-\widehat A_\alpha + \widehat A - d\widehat \theta_\alpha = 0
\end{equation}
so that $\widehat B_\alpha$ is indeed a 2-form on $\widehat Z$ and not just on the correspondence
space. Obviously, $d \widehat B_\alpha = \widehat H$. Following the descent equations one arrives at the complete
T-dual gerbe with connection, $(\widehat H, \widehat B_\alpha, \widehat F_{\alpha\beta},  (\widehat L_{\alpha\beta}, \nabla^{\widehat L_{\alpha\beta}}))$.
cf. \cite{BMPR}.

$\, $

The rules
for transforming the Ramond-Ramond (RR) fields can be encoded in the {\cite{BEM04a, BEM04b}} generalization of
{\em Hori's formula}
  {}
\begin{equation} \label{eqn:Hori}
T_*G =  \int_{Z\times_X  \widehat Z/\widehat Z} e^{ -A \wedge \widehat A }\ p^*G \,,
\end{equation}
 where $G \in \Omega^\bullet(Z)^\bbT$ is the total RR fieldstrength,
\begin{center}
$G\in\Omega^{even}(Z)^\bbT \quad$ for {   { Type IIA}};\\
$G\in\Omega^{odd}(Z)^\bbT \quad$ for {   { Type IIB}},\\
\end{center}
and where the right hand side of equation \eqref{eqn:Hori} is an invariant differential form on $Z\times_X\widehat Z$, and
the integration is along the $\bbT$-fiber of $Z$.

Recall that the twisted cohomology
is defined as the cohomology of the complex
$$H^\bullet(Z, H) = H^\bullet(\Omega^\bullet(Z), d_H=d+ H\wedge).$$
By the identity \eqref{eqn:Hori}, $T_*$ maps $d_H$-closed forms $G$ to $d_{\widehat
H}$-closed forms $T_*G$.
 So T-duality $T_*$  induces a map on twisted cohomologies,
$$
T : H^\bullet(Z, H) \to H^{\bullet +1}(\widehat Z, \widehat H).
$$
Define the Riemannian metrics on $Z$ and $\widehat Z$ respectively by
$$
g=\pi^*g_X+R^2\, A\odot A,\qquad \widehat g=\widehat\pi^*g_X+1/{R^2} \,\widehat A\odot\widehat A.
$$
where $g_X$ is a Riemannian metric on $X$.
 Then $g$ is $\TT$-invariant and the length of each circle fibre is $R$; $\widehat g$
 is $\widehat\TT$-invariant and the length of each circle fibre is $1/R$.

 The following theorem summarizes the main consequence of T-duality for principal circle bundles in a background flux.

\begin{theorem}[T-duality isomorphism \cite{BEM04a,BEM04b}]\label{thm:T-duality}
In the notation above, and
with the above choices of Riemannian metrics and flux forms, the map \eqref{eqn:Hori}
$$
T_*\colon\Omega^{\overline k}(Z)^\TT\to\Omega^{\overline{k+1}}(\widehat Z)^{\widehat\TT},
$$
for $k=0,1$, (where $\overline k$ denotes the parity of $k$) are isometries, inducing isomorphisms on twisted
cohomology groups,
\begin{equation}\label{T-duality-coh}
T : H^\bullet(Z, H) \stackrel{\cong}{\longrightarrow} H^{\bullet +1}(\widehat Z, \widehat H).
\end{equation}
Therefore under T-duality one has the exchange,
\begin{center}
{$R \Longleftrightarrow 1/R$}\quad and \quad
{{   {background H-flux} $\Longleftrightarrow$   {Chern class}}}
\end{center}
\end{theorem}

%--------------------------------------------------------------------------------------------------------------------
\subsection{The Poincar\'e Gerbe Module on the Correspondence Space and its Twisted Bismut-Chern Character} \label{Poincare}

$\, $

From Section~\ref{reviewT}, we can observe that on the correspondence space \( Z \times_X \widehat{Z} \) arising in T-duality, there appears a difference flux \( \widehat{H} -H \). Since \( [\widehat{H}] = [H] \), this difference must be the curvature of a flat gerbe, which we refer to as the \textbf{correspondence gerbe}. The correspondence space \( Z \times_X \widehat{Z} \) is a \( \TT \times \widehat{\TT} \)-bundle over \( X \), and over each point \( x \in X \), the fiber \( \TT \times \widehat{\TT} \) carries a canonical \textbf{Poincar\'e line bundle}.

In the following, we construct this correspondence gerbe as well as its holonomy line bundle explicitly and show how the family of Poincar\'e line bundles can be assembled into a gerbe module over the correspondence gerbe, which we call the \textbf{Poincar\'e Gerbe Module}. We then apply the framework reviewed in Section~\ref{exotic} to define the (horizontal) twisted Bismut--Chern character associated with the Poincar\'e Gerbe Module. This character will be essential in the next section for constructing the loop Hori formula.

\subsubsection{The correspondence gerbe on the correspondence space and its horizontal holonomy line bundle on the free loop space of the correspondence space} \label{corrgerbe}

$\, $

We adopt the same notations as in Section~\ref{reviewT}. As previously discussed, on the space \( Z \), the gerbe data is given by
\[
(H, B_\alpha, F_{\alpha\beta}, (L_{\alpha\beta}, \nabla^{L_{\alpha\beta}})),
\]
while on the dual side \( \widehat{Z} \), the corresponding dual gerbe data takes the form
\[
(\widehat{H}, \widehat{B}_\alpha, \widehat{F}_{\alpha\beta}, (\widehat{L}_{\alpha\beta}, \nabla^{\widehat{L}_{\alpha\beta}})).
\]

{\em In the present context, we work over the base manifold \( X \) using a Brylinski  cover adapted to the loop space, rather than a standard good cover. Nonetheless, all arguments from Section~\ref{reviewT} continue to apply without essential modification.}

The correspondence space \( Z \times_X \widehat{Z} \) is locally covered by open sets of the form \( U_\alpha \times \mathbb{T} \times \widehat{\mathbb{T}} \). On this space, we define the \textbf{correspondence gerbe with connection}:
\[
\left\{ {\widehat{p}}^* \widehat{L}_{\alpha\beta} \otimes (p^* L_{\alpha\beta})^{-1}, \; 1 \otimes {\widehat{p}}^* \nabla^{\widehat{L}_{\alpha\beta}} - p^* \nabla^{L_{\alpha\beta}} \otimes 1 \right\},
\]
whose curvature is clearly given by \( \widehat{H} - H \).

As introduced in~\eqref{holbundle}, we previously described the construction of the holonomy line bundle from gerbe data. In the following, we will give a horizontal version of this construction.

Let \( \mathcal{L}_h^B \) denote the \textbf{horizontal holonomy line bundle} with connection on \( LZ \), associated to the gerbe with connection \( (H, B_\alpha, F_{\alpha\beta}, (L_{\alpha\beta}, \nabla^{L_{\alpha\beta}})) \) on \( Z \). By \emph{horizontal}, we mean that the transition functions are built using the \textbf{horizontal holonomy} \( \mathrm{hol}_h(\nabla^{L_{\alpha\beta}}) \) defined with respect to the horizontal vector field \( K^h \) along loops.

The system of 1-forms \( \{ -i_{K^h} \overline{B_\alpha} \} \) satisfies the gauge transformation law:
\[
d\, \mathrm{hol}_h^{-1}(\nabla^{L_{\alpha\beta}}) = -i_{K^h} \overline{B_\alpha} + i_{K^h} \overline{B_\beta},
\]
which defines a connection \( \nabla^{\mathcal{L}_h^B} \) on \( \mathcal{L}_h^B \).

Since the flux \( H \) is \( \mathbb{T} \)-invariant, we have \( L_{K^v} \overline{H} = 0 \). Furthermore, since \( L_K \overline{H} = 0 \), it follows that \( L_{K^h} \overline{H} = 0 \) as well. Thus, we obtain a horizontal analogue of Theorem~\ref{flat}:
\be
\left( \nabla^{\mathcal{L}_h^B} - i_{K^h} + \overline{H} \right)^2 + L_{K^h} = 0.
\ee

On the dual side, we have the horizontal holonomy line bundle with connection \( (\mathcal{L}_h^{\widehat{B}}, \nabla^{\mathcal{L}_h^{\widehat{B}}}) \) over \( L\widehat{Z} \), arising from the gerbe with connection \( (\widehat{H}, \widehat{B}_\alpha, \widehat{F}_{\alpha\beta}, (\widehat{L}_{\alpha\beta}, \nabla^{\widehat{L}_{\alpha\beta}})) \) on \( \widehat{Z} \). This bundle satisfies the following horizontal flatness condition:
\[
\left( \nabla^{\mathcal{L}_h^{\widehat{B}}} - i_{\widehat{K}^h} + \overline{\widehat{H}} \right)^2 + L_{\widehat{K}^h} = 0.
\]

On the loop space of the correspondence space \( LZ \times_{LX} L\widehat{Z} \), equipped with the difference flux \( \overline{\widehat{H}} - \overline{H} \), we consider the tensor product of pullback bundles:
\[
((Lp)^*\mathcal{L}^B_h)^{-1} \otimes (L\widehat{p})^*\mathcal{L}_h^{\widehat{B}}.
\]

The \textbf{horizontal holonomy line bundle with connection of the correspondence gerbe} is then defined as:
\[
\left(((Lp)^*\mathcal{L}^B_h)^{-1} \otimes (L\widehat{p})^*\mathcal{L}_h^{\widehat{B}}, \, (Lp)^*\nabla^{(\mathcal{L}_h^B)^{-1}} \otimes 1 + 1 \otimes (L\widehat{p})^*\nabla^{\mathcal{L}_h^{\widehat{B}}}\right).
\]

This structure fits into the following diagram:
\[
\xymatrix@C=4em@R=3em{
& ((Lp)^*\mathcal{L}^B_h)^{-1} \otimes (L\widehat{p})^*\mathcal{L}_h^{\widehat{B}} \ar[d] & \\
\mathcal{L}^B_h \ar[d] 
& LZ \times_{LX} L\widehat{Z} \ar[dl]_{Lp} \ar[rd]^{L\widehat{p}} 
& \mathcal{L}_h^{\widehat{B}} \ar[d] \\
LZ \ar[rd]_{L\pi} 
& & L\widehat{Z} \ar[ld]^{L\widehat{\pi}} \\
& LX &
}
\]

Moreover, this bundle satisfies a horizontal version of Theorem~\ref{flat}, namely:
\[
\left((Lp)^*\nabla^{(\mathcal{L}_h^B)^{-1}} \otimes 1 + 1 \otimes (L\widehat{p})^*\nabla^{\mathcal{L}_h^{\widehat{B}}} - i_{K^h} + \overline{\widehat{H}} - \overline{H} \right)^2 + L_{K^h} = 0.
\]

\subsubsection{The Poincar\'e gerbe module over the correspondence gerbe}

$\, $

On the open set \( U_\alpha \times \mathbb{T} \times \widehat{\mathbb{T}} \), there exists a Poincaré line bundle \( \mathcal{P}_\alpha \), equipped with a connection \( \nabla^{\mathcal{P}_\alpha} \) whose connection 1-form, in local coordinates, is given by
\[
\frac{1}{2}\left(\theta_\alpha\, d\widehat{\theta}_\alpha - \widehat{\theta}_\alpha\, d\theta_\alpha\right).
\]
The curvature of this connection is the symplectic 2-form:
\[
d\theta_\alpha \wedge d\widehat{\theta}_\alpha.
\]

From equation~\eqref{buscher}, we have:
\[
\widehat{B}_\alpha = B_\alpha + A \wedge \widehat{A} - d\theta_\alpha \wedge d\widehat{\theta}_\alpha, \qquad
\widehat{B}_\beta = B_\beta + A \wedge \widehat{A} - d\theta_\beta \wedge d\widehat{\theta}_\beta.
\]
Hence, from~\eqref{Bfield}, it follows that:
\begin{equation} \label{curveqn}
\widehat{F}_{\alpha\beta} - F_{\alpha\beta} + d\theta_\alpha \wedge d\widehat{\theta}_\alpha = d\theta_\beta \wedge d\widehat{\theta}_\beta.
\end{equation}

Since \( H^2(U_{\alpha\beta} \times \mathbb{T} \times \widehat{\mathbb{T}}) \) has no torsion, equation~\eqref{curveqn} implies that:
\[
\mathcal{P}_\beta \otimes \mathcal{P}_\alpha^{-1} \cong \widehat{p}^*\widehat{L}_{\alpha\beta} \otimes (p^*L_{\alpha\beta})^{-1}.
\]

The 1-form
\[
\nabla^{\mathcal{P}_\beta} \otimes 1 - 1 \otimes \nabla^{\mathcal{P}_\alpha} + p^* \nabla^{L_{\alpha\beta}} \otimes 1 - 1 \otimes \widehat{p}^* \nabla^{\widehat{L}_{\alpha\beta}}
\]
is invariant under the \( L\mathbb{T} \times L\widehat{\mathbb{T}} \)-action and satisfies:
\be
\delta\left( \nabla^{\mathcal{P}_\beta}\otimes 1-1\otimes \nabla^{\mathcal{P}_\alpha}+p^* \nabla^{L_{\alpha\beta}}\otimes1 -1\otimes {\widehat p}^*\nabla^{\widehat L_{\alpha\beta}}\right)=0,  
\ee
where \( \delta \) is the Čech coboundary operator. Therefore, there exist 1-forms \( \Theta_\alpha \), \( \Theta_\beta \) such that
\[
\nabla^{\mathcal{P}_\beta} \otimes 1 - 1 \otimes \nabla^{\mathcal{P}_\alpha} + p^* \nabla^{L_{\alpha\beta}} \otimes 1 - 1 \otimes \widehat{p}^* \nabla^{\widehat{L}_{\alpha\beta}} = \Theta_\alpha - \Theta_\beta.
\]

Equation~\eqref{curveqn} then gives:
\[
d\Theta_\alpha - d\Theta_\beta = 0,
\]
so the forms \( \{ d\Theta_\alpha \} \) glue together to define a globally defined closed 2-form on \( X \), which we denote by \( \Lambda \).

Moreover, since
\[
(\nabla^{\mathcal{P}_\beta} + \Theta_\beta) \otimes 1 - 1 \otimes (\nabla^{\mathcal{P}_\alpha} + \Theta_\alpha) + p^* \nabla^{L_{\alpha\beta}} \otimes 1 - 1 \otimes \widehat{p}^* \nabla^{\widehat{L}_{\alpha\beta}} = 0,
\]
the pair \( (\mathcal{P}, \nabla^{\mathcal{P}} + \Theta) \), patched from \( \{ \mathcal{P}_\alpha, \nabla^{\mathcal{P}_\alpha} + \Theta_\alpha \} \), defines a gerbe module with connection over the correspondence gerbe:
\[
\left\{ \widehat{p}^*\widehat{L}_{\alpha\beta} \otimes (p^*L_{\alpha\beta})^{-1}, \ 1 \otimes \widehat{p}^* \nabla^{\widehat{L}_{\alpha\beta}} - p^* \nabla^{L_{\alpha\beta}} \otimes 1 \right\},
\]
whose flux is \( \widehat{H} - H \).

Recall, the horizontal holonomy line bundle with connection of the correspondence gerbe is given by:
\[
\left( ((Lp)^*\mathcal{L}^B_h)^{-1} \otimes (L\widehat{p})^*\mathcal{L}_h^{\widehat{B}}, \ (Lp)^* \nabla^{(\mathcal{L}^B_h)^{-1}} \otimes 1 + 1 \otimes (L\widehat{p})^* \nabla^{\mathcal{L}_h^{\widehat{B}}} \right).
\]

The curvature of the Poincaré gerbe module is:
\[
F_\alpha = \left( \nabla^{\mathcal{P}_\alpha} + \Theta_\alpha \right)^2 = d\theta_\alpha \wedge d\widehat{\theta}_\alpha + \Lambda.
\]

\subsubsection{The horizontal twisted Bismut-Chern character of the Poincar\'e gerbe module}

$\, $

We now turn to the horizontal version of the twisted Bismut--Chern character, i.e.,  in the construction of the Bismut--Chern character~\eqref{localBCh}, the holonomies ${}^{E_\alpha}\tau^1_0$ and ${}^{E'_\alpha}\tau^1_0$ will be replaced by their horizontal holonomies.

We define
\[
\mathrm{BCh}^h_{\widehat{H} - H}(\mathcal{P}, \nabla^{\mathcal{P}} + \Theta) \in \Omega^\bullet\left(LZ \times_{LX} L\widehat{Z},\, ((Lp)^*\mathcal{L}^B_h)^{-1} \otimes (L\widehat{p})^*\mathcal{L}_h^{\widehat{B}}\right)^{L\mathbb{T} \times L\widehat{\mathbb{T}}, K^h}
\]
as the \textbf{horizontal twisted Bismut--Chern character} of the Poincaré gerbe module \( (\mathcal{P}, \nabla^{\mathcal{P}} + \Theta) \), where $\Omega^\bullet\left(LZ \times_{LX} L\widehat{Z},\, ((Lp)^*\mathcal{L}^B_h)^{-1} \otimes (L\widehat{p})^*\mathcal{L}_h^{\widehat{B}}\right)^{L\mathbb{T} \times L\widehat{\mathbb{T}}, K^h}
$ denotes the space of $(L\mathbb{T} \times L\widehat{\mathbb{T}})$-invariant 
 differential forms $\omega$ with coefficients in  $((Lp)^*\mathcal{L}^B_h)^{-1} \otimes (L\widehat{p})^*\mathcal{L}_h^{\widehat{B}}$ on $LZ \times_{LX} L\widehat{Z}$, such that  $L_{K^h}\omega=0$. 

More explicitly, on the $\alpha$-patch of the loop space of the correspondence space, $L\pi^{-1}(U_\alpha)\times _{LU_\alpha} L\widehat\pi^{-1}(U_\alpha)$,  we have:

\begingroup
\begin{align} \label{local BCh}
&\mathrm{BCh}^h_{\widehat{H} - H, \alpha}(\mathcal{P}, \nabla^{\mathcal{P}} + \Theta)\\
=& \left( 1 + \sum_{n=1}^\infty (-u)^n 
\int\limits_{0 \leq s_1 \leq \cdots \leq s_n \leq 1} 
\widetilde{(\widehat{B}_\alpha - B_\alpha)}_{s_1} \cdots 
\widetilde{(\widehat{B}_\alpha - B_\alpha)}_{s_n} \right) \nonumber\\
&\quad \cdot \left[
\left( I + \sum_{n=1}^\infty (-u)^n 
\int\limits_{0 \leq s_1 \leq \cdots \leq s_n \leq 1} 
{}^{\mathcal{P}_\alpha} \tau^0_{s_1}(\widetilde{F_\alpha}_{s_1}) \circ \cdots \circ 
{}^{\mathcal{P}_\alpha} \tau^0_{s_n}(\widetilde{F_\alpha}_{s_n}) 
\right) \cdot {}^{\mathcal{P}_\alpha, h} \tau^1_0 \right] \nonumber\\
&\quad \otimes  \left( (Lp)^* \sigma_\alpha^{-1} \otimes (L\widehat{p})^* \widehat{\sigma}_\alpha \right) \nonumber\\
&= \left( 1 + \sum_{n=1}^\infty (-u)^n 
\int\limits_{0 \leq s_1 \leq \cdots \leq s_n \leq 1} 
\widetilde{(d\theta_\alpha \wedge d\widehat{\theta}_\alpha + \Lambda)}_{s_1} \cdots 
\widetilde{(d\theta_\alpha \wedge d\widehat{\theta}_\alpha + \Lambda)}_{s_n} \right) \nonumber\\
&\quad \cdot \left[
\left( I + \sum_{n=1}^\infty (-u)^n 
\int\limits_{0 \leq s_1 \leq \cdots \leq s_n \leq 1} 
\widetilde{(d\theta_\alpha \wedge d\widehat{\theta}_\alpha + \Lambda)}_{s_1} \cdots 
\widetilde{(d\theta_\alpha \wedge d\widehat{\theta}_\alpha + \Lambda)}_{s_n} 
\right) \cdot {}^{\mathcal{P}_\alpha, h} \tau^1_0 \right] \nonumber\\
&\quad \otimes  \left( (Lp)^* \sigma_\alpha^{-1} \otimes (L\widehat{p})^* \widehat{\sigma}_\alpha \right) \nonumber\\
&= \left[
\left( I + \sum_{n=1}^\infty (-u)^n 
\int\limits_{0 \leq s_1 \leq \cdots \leq s_n \leq 1} 
\widetilde{(A \wedge \widehat{A} + \Lambda)}_{s_1} \cdots 
\widetilde{(A \wedge \widehat{A} + \Lambda)}_{s_n} 
\right) \cdot {}^{\mathcal{P}_\alpha, h} \tau^1_0 \right] \nonumber\\
&\quad \otimes \left( (Lp)^* \sigma_\alpha^{-1} \otimes (L\widehat{p})^* \widehat{\sigma}_\alpha \right) \nonumber\\
&= \left[ e^{-u \overline{A \wedge \widehat{A} + \Lambda}} \cdot {}^{\mathcal{P}_\alpha, h} \tau^1_0 \right] \otimes \left( (Lp)^* \sigma_\alpha^{-1} \otimes (L\widehat{p})^* \widehat{\sigma}_\alpha \right). \label{eq:horizontalBCh}
\end{align}
\endgroup

By a horizontal analogue of part~(i) of Theorem~\ref{BCh prop}, we obtain the identity
\begin{equation} \label{BCh closed}
\left( (Lp)^*\nabla^{(\mathcal{L}_h^B)^{-1}} \otimes 1 
+ 1 \otimes (L\widehat{p})^* \nabla^{\mathcal{L}_h^{\widehat{B}}} 
- i_{K^h} + \overline{\widehat{H}} - \overline{H} \right) 
\mathrm{BCh}^h_{\widehat{H} - H, \alpha}(\mathcal{P}, \nabla^{\mathcal{P}} + \Theta) = 0.
\end{equation}

\smallskip

\noindent

Dually, we also have the \textbf{dual horizontal twisted Bismut--Chern character}
\[
\mathrm{BCh}^h_{H - \widehat{H}}(\mathcal{P}^{-1}, \nabla^{\mathcal{P}^{-1}} - \Theta) 
\in \Omega^\bullet\left( LZ \times_{LX} L\widehat{Z},\, (Lp)^*\mathcal{L}_h^B \otimes \left((L\widehat{p})^* \mathcal{L}_h^{\widehat{B}}\right)^{-1} \right)^{L\mathbb{T} \times L\widehat{\mathbb{T}}, K^h}.
\]

\noindent
Locally on the $\alpha$-patch, this is given by:
\begin{equation}
\begin{split}
&\mathrm{BCh}^h_{H - \widehat{H}}(\mathcal{P}^{-1}, \nabla^{\mathcal{P}^{-1}} - \Theta)\\
= {} & \left[
\left( I + \sum_{n=1}^\infty u^n \!
\int\limits_{0 \leq s_1 \leq \cdots \leq s_n \leq 1}
\widetilde{(A \wedge \widehat{A} + \Lambda)}_{s_1} \cdots 
\widetilde{(A \wedge \widehat{A} + \Lambda)}_{s_n} 
\right) \cdot {}^{\mathcal{P}_\alpha^{-1}, h} \tau^1_0 \right] \\
& \quad \otimes \left( (Lp)^* \sigma_\alpha \otimes (L\widehat{p})^* \widehat{\sigma}_\alpha^{-1} \right) \\
= {} & \left[ e^{u \, \overline{A \wedge \widehat{A} + \Lambda}} \cdot {}^{\mathcal{P}_\alpha^{-1}, h} \tau^1_0 \right] 
\otimes \left( (Lp)^* \sigma_\alpha \otimes (L\widehat{p})^* \widehat{\sigma}_\alpha^{-1} \right).
\end{split}
\end{equation}

%%%%%%%%%%%%%%%%%%%%%%%%%%%%%%%%%%%%%%%%%%%%%

\section{Loop Hori Formulae-the Loop Lifting of the Hori Map for T-duality} \label{LoopHori}

The purpose of this section is to accomplish the central objective of this paper: the construction of the \emph{loop Hori map} and the formulation of T-duality from the loop space perspective. In addition, we introduce an enhanced version, referred to as the \emph{graded loop Hori map}.

%--------------------------------------------------------------------------------------------------------------------
\subsection{ Loop Hori Formulae for T-duality}

As we have seen, looping the classical T-duality diagram:
\begin{equation} 
\label{eq:T-duality-classical}
\xymatrix @=6pc @ur {
(Z, H) \ar[d]_{\pi} &
(Z\times_X \widehat{Z}, H - \widehat{H}) \ar[d]_{\widehat{p}} \ar[l]^{p} \\
X & (\widehat{Z}, \widehat{H}) \ar[l]^{\widehat{\pi}} }
\end{equation}
we obtain the following infinite-dimensional loop space diagram:
\begin{equation}
\label{eq:T-duality-loop}
\xymatrix@C=4em@R=3em{
& ((Lp)^*\mathcal{L}^B_h)^{-1} \otimes (L\widehat{p})^*\mathcal{L}_h^{\widehat{B}} \ar[d] & \\
(\mathcal{L}^B_h, \nabla^{\mathcal{L}_h^B}) \ar[d] 
& LZ \times_{LX} L\widehat{Z} \ar[dl]_{Lp} \ar[rd]^{L\widehat{p}} 
& (\mathcal{L}_h^{\widehat{B}}, \nabla^{\mathcal{L}_h^{\widehat{B}}}) \ar[d] \\
(LZ, \overline{H}) \ar[rd]_{L\pi} 
& & (L\widehat{Z}, \overline{\widehat{H}}) \ar[ld]^{L\widehat{\pi}} \\
& LX &
}
\end{equation}

\medskip

Our objective in this section is to \textbf{lift the classical Hori map}
\begin{align}
T_* \colon \Omega^{\overline{k}}(Z)^\mathbb{T} \longrightarrow \Omega^{\overline{k+1}}(\widehat{Z})^{\widehat{\mathbb{T}}}, \\
T_* G = \int_{Z\times_X \widehat{Z}/\widehat{Z}} e^{-A \wedge \widehat{A}}\, p^*G, 
\label{eq:hori-classical}
\end{align}
to the loop space framework.

\medskip

We introduce the following notations for the loop space setup. Let
\be \Omega^\bullet(LZ, \mathcal{L}^B_h)^{L\mathbb{T}}\ee
 denote the space of smooth \( L\mathbb{T} \)-invariant differential forms on the free loop space \( LZ \), with coefficients in the horizontal holonomy line bundle \( \mathcal{L}^B_h \). 
 
 Dually, let
\be \Omega^\bullet(L\widehat{Z}, \mathcal{L}_h^{\widehat{B}})^{L\widehat{\mathbb{T}}} \ee
denote the space of smooth \( L\widehat{\mathbb{T}} \)-invariant differential forms on \( L\widehat{Z} \), with coefficients in the dual horizontal holonomy line bundle \( \mathcal{L}_h^{\widehat{B}} \).

Let \( G \in \Omega^\bullet(LZ, \mathcal{L}^B_h)^{L\mathbb{T}} \). Since the horizontal twisted Bismut--Chern character
\[
BCh^h_{\widehat{H}-H}(\mathcal{P}, \nabla^{\mathcal{P}}+\Theta)
\]
takes values in differential forms on \( LZ \times_{LX} L\widehat{Z} \), twisted by the line bundle 
\[
((Lp)^*\mathcal{L}^B_h)^{-1} \otimes (L\widehat{p})^*\mathcal{L}_h^{\widehat{B}},
\]
we consider the tensor product (using \( \widehat{\otimes} \) to denote wedge product in the differential form part):
\[
BCh^h_{\widehat{H}-H}(\mathcal{P}, \nabla^{\mathcal{P}}+\Theta) \widehat{\otimes} (Lp)^*G,
\]
which lies in
\[
\Omega^\bullet\left(LZ \times_{LX} L\widehat{Z}, (L\widehat{p})^*\mathcal{L}_h^{\widehat{B}}\right)^{L\mathbb{T} \times L\widehat{\mathbb{T}}, K^h}.
\]

Since the twisting appears only in the pullback bundle over \( L\widehat{Z} \), we can perform the integration along the fiber as defined in Definition~\ref{intfiberloop}:
\[
\int_{LZ \times_{LX} L\widehat{Z} / L\widehat{Z}} BCh^h_{\widehat{H}-H}(\mathcal{P}, \nabla^{\mathcal{P}}+\Theta) \widehat{\otimes} (Lp)^*G,
\]
which takes values in 
\[
\Omega^{\bullet+1}(L\widehat{Z}, \mathcal{L}_h^{\widehat{B}})^{L\widehat{\mathbb{T}}}.
\]
See Equations~\eqref{localtensor}–\eqref{consistency} for more technical details.

\medskip

\begin{definition} \label{loop Hori}
The \textbf{loop Hori map}
\[
LT_*: \Omega^\bullet(LZ, \mathcal{L}^B_h)^{L\mathbb{T}}[[u, u^{-1}]] \longrightarrow \Omega^{\bullet+1}(L\widehat{Z}, \mathcal{L}_h^{\widehat{B}})^{L\widehat{\mathbb{T}}}[[u, u^{-1}]]
\]
is defined by
\[
LT_* G := \int_{LZ \times_{LX} L\widehat{Z} / L\widehat{Z}} 
BCh^h_{\widehat{H}-H}(\mathcal{P}, \nabla^{\mathcal{P}}+\Theta) \widehat{\otimes} (Lp)^*G.
\]

\medskip

The \textbf{dual loop Hori map}
\[
L\widehat{T}_*: \Omega^\bullet(L\widehat{Z}, \mathcal{L}_h^{\widehat{B}})^{L\widehat{\mathbb{T}}}[[u, u^{-1}]] 
\longrightarrow \Omega^{\bullet+1}(LZ, \mathcal{L}^B_h)^{L\mathbb{T}}[[u, u^{-1}]]
\]
is defined by
\[
L\widehat{T}_* \widehat{G} := \int_{LZ \times_{LX} L\widehat{Z} / LZ} 
BCh^h_{H - \widehat{H}}(\mathcal{P}^{-1}, \nabla^{\mathcal{P}^{-1}} - \Theta) \widehat{\otimes} (L\widehat{p})^* \widehat{G}.
\]
\end{definition}

We now present the central result of this paper, expressing T-duality from the perspective of loop spaces via the \emph{loop Hori map}.

\begin{theorem}[Loop Space T-Duality] \label{main1}
\leavevmode
\begin{itemize}
    \item[(i)] The loop Hori map and its dual are inverses up to a sign:
    \[
    L\widehat{T}_* \circ LT_* = -\mathrm{Id}, \quad LT_* \circ L\widehat{T}_* = -\mathrm{Id}.
    \]
    
    \item[(ii)] Both maps are chain maps with respect to the twisted differentials. Specifically,
    \begin{align*}
        &\left( \nabla^{\mathcal{L}_h^{\widehat{B}}} - u\, i_{\widehat{K}^h} + u^{-1} \overline{\widehat{H}} \right) \circ LT_*
        = LT_* \circ \left( \nabla^{\mathcal{L}_h^B} - u\, i_{K^h} + u^{-1} \overline{H} \right), \\
        &\left( \nabla^{\mathcal{L}_h^B} - u\, i_{K^h} + u^{-1} \overline{H} \right) \circ L\widehat{T}_*
        = L\widehat{T}_* \circ \left( \nabla^{\mathcal{L}_h^{\widehat{B}}} - u\, i_{\widehat{K}^h} + u^{-1} \overline{\widehat{H}} \right).
    \end{align*}
\end{itemize}
\end{theorem}

\begin{proof} {\em For simplicity, we prove the above equalities for $u=1$.}

$\,$

{\bf (i)} Let \( G \in \Omega^\bullet(LZ, \mathcal{L}^B_h)^{L\mathbb{T}} \). Over the open set \( L\pi^{-1}(U_\alpha) \), we may write
\[
G = \omega_\alpha \otimes \sigma_\alpha, \quad \omega_\alpha \in \Omega^\bullet(L\pi^{-1}(U_\alpha))^{L\mathbb{T}}.
\]
It is not hard to see that any \( L\mathbb{T} \)-invariant vertical vector field on \( LZ \) is of the form \( f \cdot v_L \), where \( f \) is a function on \( LX \). Hence,
\[
a_\alpha = i_{v_L} \omega_\alpha
\]
is a horizontal form. Moreover,
\[
i_{v_L} \left( \omega_\alpha - \overline{A} \cdot a_\alpha \right) = 0,
\]
which implies that
\[
b_\alpha = \omega_\alpha - \overline{A} \cdot a_\alpha
\]
is also a horizontal form. In summary, we can express
\[
\omega_\alpha = b_\alpha + \overline{A} \cdot a_\alpha,
\]
where both \( a_\alpha \) and \( b_\alpha \) are horizontal forms, i.e forms on $LU_\alpha.$

Therefore on $L\pi^{-1}(U_\alpha)\times _{LU_\alpha} L\widehat\pi^{-1}(U_\alpha)$, one has
\be \label{localtensor}
\begin{split} 
&BCh^h_{\widehat H-H}(\mathcal{P}, \nabla^{\mathcal{P}}+\Theta)\widehat \otimes (Lp)^*G|_{L\pi^{-1}(U_\alpha)\times _{LU_\alpha} L\widehat\pi^{-1}(U_\alpha)}\\
=&\left(e^{-\overline{ A\wedge \widehat A+ \Lambda}}\cdot {}^{\mathcal{P}_\alpha, h} \tau^1_0
\otimes (Lp^*\sigma_\alpha^{-1} \otimes L\widehat p^*\widehat{\sigma}_\alpha)\right)\widehat \otimes \left[(b_\alpha + \overline{A} \cdot a_\alpha)\otimes Lp^*\sigma_\alpha\right]\\
=&\left[
e^{-\overline{ A\wedge \widehat A+ \Lambda}}\cdot {}^{\mathcal{P}_\alpha, h} \tau^1_0 \cdot (b_\alpha + \overline{A} \cdot a_\alpha)\right]\otimes L\widehat p^*\widehat \sigma_\alpha\\
\end{split}
\ee

Note that ${}^{\mathcal{P}_\alpha, h} \tau^1_0$ is a function on $LU_\alpha$, which satisfies
\be {}^{\mathcal{P}_\alpha, h} \tau^1_0
\otimes (Lp^*\sigma_\alpha^{-1} \otimes L\widehat p^*\widehat{\sigma}_\alpha)={}^{\mathcal{P}_\beta, h} \tau^1_0
\otimes (Lp^*\sigma_\beta^{-1} \otimes L\widehat p^*\widehat{\sigma}_\beta),\ee
and
\be (b_\alpha + \overline{A} \cdot a_\alpha)\otimes Lp^*\sigma_\alpha=(b_\beta + \overline{A} \cdot a_\beta)\otimes Lp^*\sigma_\beta.\ee
Hence after integration along the fiber, one has

\be \label{consistency} 
\begin{split}
&\left(\int_{L\pi^{-1}(U_\alpha)\times _{LU_\alpha} L\widehat\pi^{-1}(U_\alpha)/L\widehat\pi^{-1}(U_\alpha)}e^{-\overline{ A\wedge \widehat A+ \Lambda}}\cdot {}^{\mathcal{P}_\alpha, h} \tau^1_0 \cdot (b_\alpha + \overline{A} \cdot a_\alpha)\right)\otimes \widehat \sigma_\alpha\\
=&\left(\int_{L\pi^{-1}(U_\beta)\times _{LU_\beta} L\widehat\pi^{-1}(U_\beta)/L\widehat\pi^{-1}(U_\beta)}e^{-\overline{ A\wedge \widehat A+ \Lambda}}\cdot {}^{\mathcal{P}_\alpha, h} \tau^1_0 \cdot (b_\beta + \overline{A} \cdot a_\beta)\right)\otimes \widehat \sigma_\beta.
\end{split}
\ee
 
Since $\Lambda, \cdot {}^{\mathcal{P}_\alpha, h} \tau^1_0, a_\alpha, b_\alpha$ are all functions or forms on $LU_\alpha$, it is not hard to see that  
\be 
\begin{split}
&\int_{L\pi^{-1}(U_\alpha)\times _{LU_\alpha} L\widehat\pi^{-1}(U_\alpha)/L\widehat\pi^{-1}(U_\alpha)}e^{-\overline{ A\wedge \widehat A+ \Lambda}}\cdot {}^{\mathcal{P}_\alpha, h} \tau^1_0 \cdot (b_\alpha + \overline{A} \cdot a_\alpha)\\
=&i_{v_L} \left[e^{-\overline{ A\wedge \widehat A+ \Lambda}}\cdot {}^{\mathcal{P}_\alpha, h} \tau^1_0 \cdot (b_\alpha + \overline{A} \cdot a_\alpha)\right]\\
=& e^{-\overline{ A\wedge \widehat A+ \Lambda}}\cdot \left[-\overline{\widehat A}\, b_\alpha+(1-\overline{\widehat A}\, \overline A)a_\alpha\right]\cdot {}^{\mathcal{P}_\alpha, h} \tau^1_0.\\
\end{split}
\ee
Note that although the form on the above contains $A$, it is already a form on $L\widehat\pi^{-1}(U_\alpha)$. 

Hence 
\be \label{alpha1}
\begin{split}
&\left. LT_*(G)\right|_{L\widehat\pi^{-1}(U_\alpha)}\\
=&\left\{ e^{-\overline{ A\wedge \widehat A+ \Lambda}}\cdot \left[-\overline{\widehat A}\, b_\alpha+(1-\overline{\widehat A}\, \overline A)a_\alpha\right]\cdot {}^{\mathcal{P}_\alpha, h} \tau^1_0\right\}\otimes \widehat \sigma_\alpha.
\end{split}
 \ee

Now considering $L{\widehat T}_* (LT_*(G))$. 
It is not hard to see that on the $\alpha$-chart, this amounts to compute the integration along the fiber $\int_{L\pi^{-1}(U_\alpha)\times _{LU_\alpha} L\widehat\pi^{-1}(U_\alpha)/\pi^{-1}(U_\alpha)}$ of the following term
{\small
\be \label{alpha11}
\begin{split}
&BCh^h_{H-\widehat H}(\mathcal{P}^{-1}, \nabla^{\mathcal{P}^{-1}}-\Theta)\cdot \left\{ e^{-\overline{ A\wedge \widehat A+ \Lambda}} \left[-\overline{\widehat A}\, b_\alpha+(1-\overline{\widehat A}\, \overline A)a_\alpha\right]\cdot {}^{\mathcal{P}_\alpha, h} \tau^1_0\right\}\otimes L\widehat p^*\widehat \sigma_\alpha\\
=& \left[
e^{-\overline{ A\wedge \widehat A+ \Lambda}}\cdot {}^{\mathcal{P}_\alpha^{-1}, h} \tau^1_0 \right] 
\otimes (Lp^*\sigma_\alpha \otimes L\widehat p^*\widehat{\sigma}_\alpha^{-1})\widehat \otimes \left\{ e^{-\overline{ A\wedge \widehat A+ \Lambda}} \cdot  \left[-\overline{\widehat A}\, b_\alpha+(1-\overline{\widehat A}\, \overline A)a_\alpha\right]\cdot {}^{\mathcal{P}_\alpha, h} \tau^1_0\right\}\otimes L\widehat p^*\widehat \sigma_\alpha\\
=&\left[-\overline{\widehat A}\, b_\alpha+(1-\overline{\widehat A}\, \overline A)a_\alpha\right]\otimes Lp^*\sigma_\alpha
\end{split}
 \ee}
 
 But 
 \be \label{alpha2}
\begin{split}
&\int_{L\pi^{-1}(U_\alpha)\times _{LU_\alpha} L\widehat\pi^{-1}(U_\alpha)/L\pi^{-1}(U_\alpha)}\left[-\overline{\widehat A}\, b_\alpha+(1-\overline{\widehat A}\, \overline A)a_\alpha\right]\\
=&i_{\widehat v_{L}}\left[-\overline{\widehat A}\, b_\alpha+(1-\overline{\widehat A}\, \overline A)a_\alpha\right]\\
=&-b_\alpha-\overline A\, a_\alpha\\
=&-\omega_\alpha.
\end{split} \ee

Combining (\ref{alpha1}) and (\ref{alpha11}), one can see that on for $G|_{ L\pi^{-1}(U_\alpha)}$, 
\h
BCh^h_{H-\widehat H}(\mathcal{P}^{-1}, \nabla^{\mathcal{P}^{-1}}-\Theta)\cdot (L\widehat{p})^*(LT_*(G))=\left[-\overline{\widehat A}\, b_\alpha+(1-\overline{\widehat A}\, \overline A)a_\alpha\right]\otimes Lp^*\sigma_\alpha
\e
Then (\ref{alpha2}) tells us that 
\h 
\begin{split}
&\int_{L\pi^{-1}(U_\alpha)\times _{LU_\alpha} L\widehat\pi^{-1}(U_\alpha)/L\pi^{-1}(U_\alpha)} BCh^h_{H-\widehat H}(\mathcal{P}^{-1}, \nabla^{\mathcal{P}^{-1}}-\Theta)\cdot (L\widehat{p})^*(LT_*(G))\\
=&- \omega_\alpha \otimes \sigma_\alpha\\
=&-G|_{ L\pi^{-1}(U_\alpha)}.
\end{split}
\e

This shows that 
\be  L{\widehat T}_*( LT_*(G))=-G.\ee

So $L{\widehat T}_*\circ LT_*=-\mathrm{Id}$. And similarly, we can prove that $LT_*\circ L{\widehat T}_*=-\mathrm{Id}.$

$\, $

{\bf (ii)} 
On the $\alpha$-patch, we have 
\be \label{localdiff}
\begin{split}
&\left(\nabla^{\cL_h^{\widehat B}}-i_{{\widehat K}^h}+\overline {\widehat H}\right)(LT_*(G))\\
=& \left(\nabla^{\cL_h^{\widehat B}}-i_{{\widehat K}^h}+\overline {\widehat H}\right)\left(\int_{ LZ \times_{LX} L\widehat{Z}/L\widehat{Z}} BCh^h_{\widehat H-H}(\mathcal{P}, \nabla^{\mathcal{P}}\Theta)\widehat \otimes (Lp)^*G\right)\\
& \left(\nabla^{\cL_h^{\widehat B}}-i_{{\widehat K}^h}+\overline {\widehat H}\right)\left\{\left(\int_{L\pi^{-1}(U_\alpha)\times _{LU_\alpha} L\widehat\pi^{-1}(U_\alpha)/L\widehat\pi^{-1}(U_\alpha)}e^{-\overline{ A\wedge \widehat A+ \Lambda}}\cdot {}^{\mathcal{P}_\alpha, h} \tau^1_0 \cdot \omega_\alpha\right)\otimes \widehat \sigma_\alpha\right\}\\
=& \left\{(d-i_{\widehat K^h}\overline {\widehat B_\alpha}-i_{{\widehat K}^h}+\overline {\widehat H})\left(\int_{L\pi^{-1}(U_\alpha)\times _{LU_\alpha} L\widehat\pi^{-1}(U_\alpha)/L\widehat\pi^{-1}(U_\alpha)}e^{-\overline{ A\wedge \widehat A+ \Lambda}}\cdot {}^{\mathcal{P}_\alpha, h} \tau^1_0 \cdot \omega_\alpha\right)\right\} \otimes \widehat \sigma_\alpha.
\end{split}
\ee

But by (\ref{prop-d}),  
\be \label{d}
\begin{split}
&d \int_{L\pi^{-1}(U_\alpha)\times _{LU_\alpha} L\widehat\pi^{-1}(U_\alpha)/L\widehat\pi^{-1}(U_\alpha)}e^{-\overline{ A\wedge \widehat A+ \Lambda}}\cdot {}^{\mathcal{P}_\alpha, h} \tau^1_0 \cdot \omega_\alpha\\
=&\int_{L\pi^{-1}(U_\alpha)\times _{LU_\alpha} L\widehat\pi^{-1}(U_\alpha)/L\widehat\pi^{-1}(U_\alpha)} -d\left(e^{-\overline{ A\wedge \widehat A+ \Lambda}}\cdot {}^{\mathcal{P}_\alpha, h} \tau^1_0 \cdot \omega_\alpha \right)\\
=&\int_{L\pi^{-1}(U_\alpha)\times _{LU_\alpha} L\widehat\pi^{-1}(U_\alpha)/L\widehat\pi^{-1}(U_\alpha)} -d\left(e^{-\overline{ A\wedge \widehat A+ \Lambda}}\cdot {}^{\mathcal{P}_\alpha, h} \tau^1_0\right) \cdot \omega_\alpha-e^{-\overline{ A\wedge \widehat A+ \Lambda}}\cdot {}^{\mathcal{P}_\alpha, h} \tau^1_0\cdot d\omega_\alpha.
\end{split}
\ee

By (\ref{prop-i}), we have 
\be \label{i}
\begin{split}
&-i_{\widehat K^h} \int_{L\pi^{-1}(U_\alpha)\times _{LU_\alpha} L\widehat\pi^{-1}(U_\alpha)/L\widehat\pi^{-1}(U_\alpha)}e^{-\overline{ A\wedge \widehat A+ \Lambda}}\cdot {}^{\mathcal{P}_\alpha, h} \tau^1_0 \cdot \omega_\alpha\\
=& \int_{L\pi^{-1}(U_\alpha)\times _{LU_\alpha} L\widehat\pi^{-1}(U_\alpha)/L\widehat\pi^{-1}(U_\alpha)} i_{K^h}(e^{-\overline{ A\wedge \widehat A+ \Lambda}}\cdot {}^{\mathcal{P}_\alpha, h} \tau^1_0)\cdot \omega_\alpha+e^{-\overline{ A\wedge \widehat A+ \Lambda}}\cdot {}^{\mathcal{P}_\alpha, h} \tau^1_0 \cdot i_{K^h}\omega_\alpha.
\end{split}
\ee

However, by (\ref{BCh closed}), in the $\alpha$-patch, we have
\be (d+i_{K^h}\overline {B_\alpha}-i_{\widehat K^h}\overline {\widehat B_\alpha}-i_{K^h}+\overline{\widehat H}-\overline H)(e^{-\overline{ A\wedge \widehat A+ \Lambda}}\cdot {}^{\mathcal{P}_\alpha, h} \tau^1_0)=0. \ee

Hence, combing (\ref{d}) and (\ref{i}), we see that 
\be
\begin{split}
&(d-i_{\widehat K^h}\overline {\widehat B_\alpha}-i_{{\widehat K}^h}+\overline {\widehat H})\left(\int_{L\pi^{-1}(U_\alpha)\times _{LU_\alpha} L\widehat\pi^{-1}(U_\alpha)/L\widehat\pi^{-1}(U_\alpha)}e^{-\overline{ A\wedge \widehat A+ \Lambda}}\cdot {}^{\mathcal{P}_\alpha, h} \tau^1_0 \cdot \omega_\alpha\right)\\
=&(d-i_{{\widehat K}^h}) \left(\int_{L\pi^{-1}(U_\alpha)\times _{LU_\alpha} L\widehat\pi^{-1}(U_\alpha)/L\widehat\pi^{-1}(U_\alpha)}e^{-\overline{ A\wedge \widehat A+ \Lambda}}\cdot {}^{\mathcal{P}_\alpha, h} \tau^1_0 \cdot \omega_\alpha\right)\\
&+(-i_{\widehat K^h}\overline {\widehat B_\alpha}+\overline {\widehat H})\left(\int_{L\pi^{-1}(U_\alpha)\times _{LU_\alpha} L\widehat\pi^{-1}(U_\alpha)/L\widehat\pi^{-1}(U_\alpha)}e^{-\overline{ A\wedge \widehat A+ \Lambda}}\cdot {}^{\mathcal{P}_\alpha, h} \tau^1_0 \cdot \omega_\alpha\right)\\
=&\int_{L\pi^{-1}(U_\alpha)\times _{LU_\alpha} L\widehat\pi^{-1}(U_\alpha)/L\widehat\pi^{-1}(U_\alpha)}(-d+i_{K^h})\left(e^{-\overline{ A\wedge \widehat A+ \Lambda}}\cdot {}^{\mathcal{P}_\alpha, h} \tau^1_0\right)\cdot \omega_\alpha-e^{-\overline{ A\wedge \widehat A+ \Lambda}}\cdot {}^{\mathcal{P}_\alpha, h} \tau^1_0\cdot(d-i_{K^h})\omega_\alpha\\
&+\int_{L\pi^{-1}(U_\alpha)\times _{LU_\alpha} L\widehat\pi^{-1}(U_\alpha)/L\widehat\pi^{-1}(U_\alpha)}(i_{\widehat K^h}\overline {\widehat B_\alpha}-\overline {\widehat H})e^{-\overline{ A\wedge \widehat A+ \Lambda}}\cdot {}^{\mathcal{P}_\alpha, h} \tau^1_0 \cdot \omega_\alpha\\
=&\int_{L\pi^{-1}(U_\alpha)\times _{LU_\alpha} L\widehat\pi^{-1}(U_\alpha)/L\widehat\pi^{-1}(U_\alpha)} \left(e^{-\overline{ A\wedge \widehat A+ \Lambda}}\cdot {}^{\mathcal{P}_\alpha, h} \tau^1_0\right) \cdot (d-i_{K^h}\overline {B_\alpha}-i_{K^h}+\overline H)\omega_\alpha\\
=&\int_{L\pi^{-1}(U_\alpha)\times _{LU_\alpha} L\widehat\pi^{-1}(U_\alpha)/L\widehat\pi^{-1}(U_\alpha)} \left(e^{-\overline{ A\wedge \widehat A+ \Lambda}}\cdot {}^{\mathcal{P}_\alpha, h} \tau^1_0\right) \cdot (d-i_{K^h}\overline {B_\alpha}-i_{K^h}+\overline H)\omega_\alpha.
\end{split}
\ee

Note that in the $\alpha$-patch, 
\be  (\nabla^{\cL_h^B}-i_{K^h}+\overline H)G=\left\{(d-i_{K^h}\overline {B_\alpha}-i_{K^h}+\overline H)\omega_\alpha\right\}\otimes \sigma_\alpha.\ee

Hence continuing (\ref{localdiff}), we have in the $\alpha$-patch
\be 
\begin{split}
&\left(\nabla^{\cL_h^{\widehat B}}-i_{{\widehat K}^h}+\overline {\widehat H}\right)\circ LT_*(G)\\
=& \left\{(d-i_{\widehat K^h}\overline {\widehat B_\alpha}-i_{{\widehat K}^h}+\overline {\widehat H})\left(\int_{L\pi^{-1}(U_\alpha)\times _{LU_\alpha} L\widehat\pi^{-1}(U_\alpha)/L\widehat\pi^{-1}(U_\alpha)}e^{-\overline{ A\wedge \widehat A+ \Lambda}}\cdot {}^{\mathcal{P}_\alpha, h} \tau^1_0 \cdot \omega_\alpha\right)\right\} \otimes \widehat \sigma_\alpha\\
=& \int_{ LZ \times_{LX} L\widehat{Z}/L\widehat{Z}} BCh^h_{\widehat H-H}(\mathcal{P}, \nabla^{\mathcal{P}}\Theta)\widehat \otimes  (\nabla^{\cL_h^B}-i_{K^h}+\overline H)G\\
=&LT_* \left( (\nabla^{\cL_h^B}-i_{K^h}+\overline H)G\right).
\end{split}
\ee
Hence 
$$\left(\nabla^{\cL_h^{\widehat B}}-ui_{{\widehat K}^h}+u^{-1}\overline {\widehat H}\right)\circ LT_*=LT_*\circ (\nabla^{\cL_h^B}-ui_{K^h}+u^{-1}\overline H). $$
Similarly, we can also prove 
\be (\nabla^{\cL_h^B}-ui_{K^h}+u^{-1}\overline H)\circ L\widehat{T}_*=L\widehat{T}_*\circ  \left(\nabla^{\cL_h^{\widehat B}}-ui_{{\widehat K}^h}+u^{-1}\overline {\widehat H}\right). \ee

\end{proof}

As a consequence of Theorem~\ref{main1},
\begin{theorem}\label{main2}
The loop Hori map
\begin{footnotesize}
\[
LT_* \colon 
\left(
\Omega^\bullet(LZ, \mathcal{L}^B_h)^{L\mathbb{T}, K^h}[[u, u^{-1}]],\ 
\nabla^{\mathcal{L}_h^B} - u\, i_{K^h} + u^{-1} \overline{H}
\right)
\longrightarrow 
\left(
\Omega^{\bullet+1}(L\widehat{Z}, \mathcal{L}_h^{\widehat{B}})^{L\widehat{\mathbb{T}}, \widehat{K}^h}[[u, u^{-1}]],\ 
\nabla^{\mathcal{L}_h^{\widehat{B}}} - u\, i_{\widehat{K}^h} + u^{-1} \overline{\widehat{H}}
\right)
\]
\end{footnotesize}
is a quasi-isomorphism.
\end{theorem}

%--------------------------------------------------------------------------------------------------------------------
\subsection{Restriction to Contant Loop Spaces: from Loop Hori to Hori}

$\, $

Restricting the loop Hori map to the constant loop spaces, we obtain a slight modified version of the classical Hori map~\eqref{eqn:Hori}. 

\begin{definition}
The \textbf{modified Hori map}
\[
T_*': \Omega^\bullet(Z)^{\mathbb{T}}[[u, u^{-1}]] \longrightarrow \Omega^{\bullet+1}(\widehat{Z})^{\widehat{\mathbb{T}}}[[u, u^{-1}]]
\]
is defined by
\[
T_*'(G) := \int_{Z \times_X \widehat{Z} / \widehat{Z}} e^{-u(A \wedge \widehat{A} + \Lambda)} \cdot p^*G.
\]

The \textbf{dual modified Hori map}
\[
\widehat{T}_*': \Omega^\bullet(\widehat{Z})^{\widehat{\mathbb{T}}}[[u, u^{-1}]] \longrightarrow \Omega^{\bullet+1}(Z)^{\mathbb{T}}[[u, u^{-1}]]
\]
is defined by
\[
\widehat{T}_*'(\widehat{G}) := \int_{\widehat{Z} \times_X Z / Z} e^{u(A \wedge \widehat{A} + \Lambda)} \cdot \widehat{p}^*(\widehat{G}).
\]
\end{definition}

Let
\[
h_{L\mathbb{T}, K^h}^\bullet(LZ, \mathcal{L}^B_h; \overline{H}) := H\left(
\Omega^\bullet(LZ, \mathcal{L}^B_h)^{L\mathbb{T}, K^h}[[u, u^{-1}]],\ 
\nabla^{\mathcal{L}^B_h} - u i_{K^h} + u^{-1} \overline{H}
\right)
\]
denote the cohomology of the horizontal exotic twisted $L\mathbb{T}$-equivariant complex.

{\bf Since \( \Lambda \) is a closed 2-form on \( X \), it follows that the modified Hori maps \( T_*' \) and \( \widehat{T}_*' \) enjoy the same formal properties as their unmodified counterparts described in Theorem~\ref{thm:T-duality}.}

\medskip

By Theorem~\ref{local}, the restriction to constant loops induces an isomorphism:
\[
\mathrm{res}: h_{L\mathbb{T}, K^h}^\bullet(LZ, \mathcal{L}^B_h; \overline{H}) 
\xrightarrow{\cong} 
H^\bullet(\Omega(Z)^{\mathbb{T}}[[u, u^{-1}]],\ d + u^{-1}H),
\]
and similarly on the dual side.

\begin{theorem} \label{main3}
There is a commutative diagram:
\begin{equation} \label{T-duality-commute}
\begin{gathered}
\xymatrix@C=4.5em@R=3em{
h_{L\mathbb{T}, K^h}^\bullet(LZ, \mathcal{L}^B_h; \overline{H}) 
\ar[d]_{\mathrm{res}}^{\cong}  
\ar[r]^{LT_*} & 
h_{L\widehat{\mathbb{T}}, \widehat{K}^h}^{\bullet+1}(L\widehat{Z}, \mathcal{L}_h^{\widehat{B}}; \overline{\widehat{H}}) 
\ar[d]_{\mathrm{res}}^{\cong} \\
H^\bullet(\Omega(Z)^{\mathbb{T}}[[u, u^{-1}]], d + u^{-1}H) 
\ar[r]_{T_*'} & 
H^{\bullet+1}(\Omega(\widehat{Z})^{\widehat{\mathbb{T}}}[[u, u^{-1}]], d + u^{-1} \widehat{H})
}
\end{gathered}
\end{equation}
\end{theorem}

$\, $
%--------------------------------------------------------------------------------------------------------------------
\subsection{ The Graded Loop Hori Formulae} \label{graded}

$\, $

Motivated by the graded Hori formula introduced in~\cite{HM21}, we define a loop space analogue—the \emph{graded loop Hori map}—in this subsection.

We begin by recalling the classical setup. For each \( m \in \mathbb{Z} \), define the \textbf{level-$m$ Hori map} by
\[
T_{*, m}(G) := \int_{\mathbb{T}} e^{-m A \wedge \widehat{A}} \cdot G.
\]
Assembling these level-wise maps yields the \textbf{graded Hori map}
\[
gT_* \colon \bigoplus_{m \in \mathbb{Z}} \Omega^\bullet(Z)^{\mathbb{T}}[[u, u^{-1}]] \cdot y^m 
\longrightarrow \bigoplus_{m \in \mathbb{Z}} \Omega^{\bullet+1}(\widehat{Z})^{\widehat{\mathbb{T}}}[[u, u^{-1}]] \cdot y^m,
\]
defined by
\[
gT_*\left( \sum_{m \in \mathbb{Z}} \omega_m \cdot y^m \right) := \sum_{m \in \mathbb{Z}} T_{*, m}(\omega_m) \cdot y^m,
\]
where \( y \) is a formal grading variable. A dual map \( g\widehat{T}_* \) can be defined analogously on the dual side. The key identity shown in~\cite{HM21} is:
\[
g\widehat{T}_* \circ gT_* = -y \frac{\partial}{\partial y}, \qquad
gT_* \circ g\widehat{T}_* = -y \frac{\partial}{\partial y}.
\]

\begin{definition}
The \textbf{level-$m$ loop Hori map}
\[
LT_{*, m} \colon \Omega^\bullet(LZ, (\mathcal{L}_h^B)^{\otimes m})^{L\mathbb{T}}[[u, u^{-1}]]
\longrightarrow \Omega^{\bullet+1}(L\widehat{Z}, (\mathcal{L}_h^{\widehat{B}})^{\otimes m})^{L\widehat{\mathbb{T}}}[[u, u^{-1}]]
\]
is defined by
\[
LT_{*, m}(G) := \int_{LZ \times_{LX} L\widehat{Z} / L\widehat{Z}} 
BCh^h_{m(\widehat{H} - H)}\big(\mathcal{P}^{\otimes m}, \nabla^{\mathcal{P}^{\otimes m}} + m \Theta\big) 
\widehat{\otimes} (Lp)^*G.
\]

The \textbf{dual level-$m$ loop Hori map}
\[
L\widehat{T}_{*, m} \colon \Omega^\bullet(L\widehat{Z}, (\mathcal{L}_h^{\widehat{B}})^{\otimes m})^{L\widehat{\mathbb{T}}}[[u, u^{-1}]]
\longrightarrow \Omega^{\bullet+1}(LZ, (\mathcal{L}_h^B)^{\otimes m})^{L\mathbb{T}}[[u, u^{-1}]]
\]
is defined by
\[
L\widehat{T}_{*, m}(\widehat{G}) := \int_{LZ \times_{LX} L\widehat{Z} / LZ}
BCh^h_{m(H - \widehat{H})}\big((\mathcal{P}^{-1})^{\otimes m}, \nabla^{(\mathcal{P}^{-1})^{\otimes m}} - m \Theta\big) 
\widehat{\otimes} (L\widehat{p})^*(\widehat{G}).
\]

Assembling these maps for all \( m \in \mathbb{Z} \), we define the \textbf{graded loop Hori map}
\[
gLT_* \colon \bigoplus_{m \in \mathbb{Z}} \Omega^\bullet(LZ, (\mathcal{L}_h^B)^{\otimes m})^{L\mathbb{T}}[[u, u^{-1}]] \cdot y^m
\longrightarrow \bigoplus_{m \in \mathbb{Z}} \Omega^{\bullet+1}(L\widehat{Z}, (\mathcal{L}_h^{\widehat{B}})^{\otimes m})^{L\widehat{\mathbb{T}}}[[u, u^{-1}]] \cdot y^m
\]
via
\[
gLT_*\left( \sum_{m \in \mathbb{Z}} \omega_m \cdot y^m \right) := \sum_{m \in \mathbb{Z}} LT_{*, m}(\omega_m) \cdot y^m.
\]

The \textbf{dual graded loop Hori map} \( gL\widehat{T}_* \) is defined analogously.
\end{definition}

\begin{theorem}
The graded loop Hori maps satisfy the identities:
\[
gL\widehat{T}_* \circ gLT_* = -y \frac{\partial}{\partial y}, \qquad 
gLT_* \circ gL\widehat{T}_* = -y \frac{\partial}{\partial y}.
\]
\end{theorem}
\begin{proof} The proof is similar to the proof of (i) in Theorem \ref{main1}. 

Take a 
$$G\in \Omega^\bullet(LZ, (\mathcal{L}^B_h)^{\otimes m})^{L\mathbb{T}}[[u, u^{-1}]]. $$ Suppose in the $\alpha$-patch, 
\be  G = \omega_\alpha \otimes (\sigma_\alpha)^{\otimes m}=(b_\alpha + \overline{A} \cdot a_\alpha)\otimes (\sigma_\alpha)^{\otimes m}.\ee

Then similar to (\ref{alpha1}), one has
\be 
\begin{split}
&\left. LT_{*, m}(G)\right|_{L\widehat\pi^{-1}(U_\alpha)}\\
=&\left\{ e^{-m\overline{ A\wedge \widehat A+ \Lambda}}\cdot \left[-m\overline{\widehat A}\, b_\alpha+(1-m\overline{\widehat A}\, \overline A)a_\alpha\right]\cdot {}^{\mathcal{P}_\alpha, h} \tau^1_0\right\}\otimes \widehat \sigma_\alpha.
\end{split}
 \ee
And similar to (\ref{alpha2}), one has
\be 
\begin{split}
&\int_{L\pi^{-1}(U_\alpha)\times _{LU_\alpha} L\widehat\pi^{-1}(U_\alpha)/L\pi^{-1}(U_\alpha)}\left[-m\overline{\widehat A}\, b_\alpha+(1-m\overline{\widehat A}\, \overline A)a_\alpha\right]\\
=&i_{\widehat v_{L}}\left[-m\overline{\widehat A}\, b_\alpha+(1-m\overline{\widehat A}\, \overline A)a_\alpha\right]\\
=&-mb_\alpha-m\overline A\, a_\alpha\\
=&-m\omega_\alpha.
\end{split} \ee
Hence one has 
\be L\widehat T_{*, m}\circ LT_{*, m}(G)=-mG.\ee

Then the desired equality follows. 
\end{proof}

%--------------------------------------------------------------------------------------------------------------------
\subsection{ An Example} \label{example}

We consider the simplest situation,
\[
  Z = X \times \mathbb T, 
  \qquad 
  \widehat Z = X \times \widehat{\mathbb T},
\]
with trivial gerbe data on both sides. Thus on $Z$ we take
\[
  \bigl(H,\, B_\alpha,\, F_{\alpha\beta},\, (L_{\alpha\beta}, \nabla^{L_{\alpha\beta}})\bigr)
  = (0,0,0,\text{trivial}),
\]
and similarly on $\widehat Z$,
\[
  \bigl(\widehat H,\, \widehat B_\alpha,\, \widehat F_{\alpha\beta},\, 
        (\widehat L_{\alpha\beta}, \nabla^{\widehat L_{\alpha\beta}})\bigr)
  = (0,0,0,\text{trivial}).
\]
Write $\theta$ (resp.\ $\widehat\theta$) for the angular coordinate on the
$\mathbb T$– (resp.\ $\widehat{\mathbb T}$–)factor; the canonical connection
$1$–forms are $d\theta$ and $d\widehat\theta$ on $Z$ and $\hat Z$ respectively.

\medskip

Let $p\colon Z\to X$ and $\widehat p\colon \widehat Z\to X$ be the projections,
and denote their loop maps by $Lp$ and $L\widehat p$.  In
Definition~\ref{loop Hori}, the loop Hori map
\[
  LT_*:\;
  \Omega^\bullet\!\bigl(LZ, \mathcal{L}^{B}_{h}\bigr)^{L\mathbb T}
  \bigl[\![u, u^{-1}]\!\bigr]
  \;\longrightarrow\;
  \Omega^{\bullet+1}\!\bigl(L\widehat Z, \mathcal{L}^{\widehat B}_{h}\bigr)^{L\widehat{\mathbb T}}
  \bigl[\![u, u^{-1}]\!\bigr]
\]
reduces, by triviality of the twists, to the untwisted identifications
\begin{align*}
  \Omega^\bullet\!\bigl(LZ, \mathcal{L}^{B}_{h}\bigr)^{L\mathbb T}\bigl[\![u,u^{-1}]\!\bigr]
  &= \Omega^\bullet\!\bigl(LX \times L\mathbb T\bigr)^{L\mathbb T}\bigl[\![u,u^{-1}]\!\bigr],\\[2pt]
  \Omega^{\bullet+1}\!\bigl(L\widehat Z, \mathcal{L}^{\widehat B}_{h}\bigr)^{L\widehat{\mathbb T}}\bigl[\![u,u^{-1}]\!\bigr]
  &= \Omega^{\bullet+1}\!\bigl(LX \times L\widehat{\mathbb T}\bigr)^{L\widehat{\mathbb T}}\bigl[\![u,u^{-1}]\!\bigr].
\end{align*}

For the twisted Bismut–Chern character (see~\eqref{local BCh}), the present
case leads to a particularly transparent factor:
\begin{equation}\label{eq:BCh-trivial}
  BCh^{h}_{\widehat H - H}\!\bigl(\mathcal P,\, \nabla^{\mathcal P}+\Theta\bigr)
  \;=\;
  \exp\!\bigl(-\,u\, \overline{\,d\theta \wedge d\widehat\theta\,}\bigr),
\end{equation}
where $u$ is the usual formal parameter of degree $2$, and the bar denotes the
form induced on
$L(X\times \mathbb T \times \widehat{\mathbb T})$ by the average operation (see section (\ref{reviewExotic})). 

\medskip

Accordingly, the loop Hori map takes the simple form
\[
  LT_*:\;
  \Omega^\bullet\!\bigl(LX \times L\mathbb T\bigr)^{L\mathbb T}\bigl[\![u,u^{-1}]\!\bigr]
  \;\longrightarrow\;
  \Omega^{\bullet+1}\!\bigl(LX \times L\widehat{\mathbb T}\bigr)^{L\widehat{\mathbb T}}\bigl[\![u,u^{-1}]\!\bigr],
\]
and is given explicitly by fibre integration along the correspondence
$LZ \times_{LX} L\widehat Z$:
\[
  LT_*\,G
  \;=\;
  \int_{LZ \times_{LX} L\widehat Z\,/\,L\widehat Z}
  \exp\!\bigl(-\,u\, \overline{\,d\theta \wedge d\widehat\theta\,}\bigr)
  \,\wedge\, (Lp)^{*}G,
  \qquad G\in \Omega^\bullet\!\bigl(LX \times L\mathbb T\bigr)^{L\mathbb T}\bigl[\![u,u^{-1}]\!\bigr].
\]

Dually, one has the dual loop Hori map
\[
  \widehat{LT}_*:\;
  \Omega^\bullet\!\bigl(LX \times L\widehat{\mathbb T}\bigr)^{L\widehat{\mathbb T}}\bigl[\![u,u^{-1}]\!\bigr]
  \;\longrightarrow\;
  \Omega^{\bullet+1}\!\bigl(LX \times L\mathbb T\bigr)^{L\mathbb T}\bigl[\![u,u^{-1}]\!\bigr],
\]
defined by fibre integration in the opposite direction with the conjugate kernel:
\begin{equation}\label{eq:dual-loop-hori}
  \widehat{LT}_*\,\widehat G
  \;=\;
  \int_{LZ \times_{LX} L\widehat Z\,/\,LZ}
  \exp\!\bigl(u\, \overline{\,d\theta \wedge d\widehat\theta\,}\bigr)
  \,\wedge\, (L\widehat p)^{*}\widehat G,
  \qquad 
  \widehat G\in \Omega^\bullet\!\bigl(LX \times L\widehat{\mathbb T}\bigr)^{L\widehat{\mathbb T}} \bigl[\![u,u^{-1}]\!\bigr].
\end{equation}

\medskip

In the present trivial model the two transformations are mutual inverses
in the sense that the compositions
\[
  \widehat{LT}_*\circ LT_*
  \quad\text{and}\quad
  LT_* \circ \widehat{LT}_*
\]
act as minus identity on the respective spaces of $L\mathbb T$- and
$L\widehat{\mathbb T}$-invariant forms, reflecting the expected involutivity of T-duality at the level of
loop-space differential forms.

%%%%%%%%%%%%%%%%%%%%%%%%%%%%%%%%%%%%%%%%%%%%%%%%%%%%%%%%%%%%%%%

\end{document}